\documentclass[a4paper,reqno,10pt,oneside]{amsart}

\usepackage{latexsym}
\usepackage{float}
\usepackage{amsfonts,amssymb,latexsym,xspace,epsfig,graphics,color}
\usepackage{amsmath,enumerate,stmaryrd,xy,stackrel}
\usepackage[footnotesize]{caption}
\usepackage{indentfirst}
\usepackage[T1]{fontenc}
\usepackage{a4wide}
\usepackage{url}
\usepackage{tikz}
\usepackage{pgfplots}
\usepackage{colortbl}
\usepgfplotslibrary{fillbetween}
\usepackage{color}
\usepackage{graphpap}
\usepackage{epsfig}
\usepackage{psfrag}
\usepackage{graphicx}
\usepackage{subfigure}
\usetikzlibrary{arrows,snakes,backgrounds,positioning}
\usepackage{tikz-network}
\usepackage{multirow}
\usepackage{float}
\usepackage{rotating}
\usepackage{enumitem}
\usepackage{booktabs}
\theoremstyle{plain}
\newtheorem{thm}{Theorem}[section]
\newtheorem{cor}{Corollary}[section]

\newtheorem{prop}{Proposition}[section]

\theoremstyle{definition}

\theoremstyle{remark}

\numberwithin{equation}{section}
\numberwithin{figure}{section}

\begin{document}

\title[]{The impact of effective participation in stopping misinformation: an approach based on branching processes}

\date{}

\author[Luz Marina Gomez]{Luz Marina Gomez}  
\address{Luz Marina Gomez. Universidade de São Paulo, São Paulo, SP, Brazil. E-mail: luzgomez@alumni.usp.br}

\author[Valdivino V. Junior]{Valdivino V. Junior}  
\address{Valdivino Vargas. Universidade Federal de Goias, Campus Samambaia, CEP 74001-970, Goi\^ania, GO, Brazil. E-mail: vvjunior@ufg.br}

\author[Pablo Rodriguez]{Pablo M. Rodriguez} 
\address{Pablo Rodriguez. Universidade Federal de Pernambuco, Av. Prof. Moraes Rego, 1235. Cidade Universit\'aria, CEP 50670-901, Recife, PE, Brazil. E-mail: pablo@de.ufpe.br}



\maketitle

\begin{abstract}
The emergence of research focused to understand the spreading and impact of disinformation is increasing year over year. Most times, the purpose of those who start the spreading of information intentionally false and designed to cause harm is in catalyzing its fast transformation into misinformation, which is the false content shared by people who do not realize it is false or misleading. Our interest is in discussing the role of people who decide to adopt an active role in stopping the propagation of an information when they realize that it is false. For this, we formulate two simple probabilistic models to compare misinformation spreading in the possible scenarios for which there is a passive or an active environment of aware individuals. With aware individuals we mean those individuals who realize that a given information is false or misleading. In the passive environment we assume that if one of an aware individual is exposed to the misinformation then he/she will not spread it. In the active environment we assume that if one of an aware individual is exposed to the misinformation then he/she will not spread it but also he/she will stop the propagation to other individuals from the individual who contacted him/her. We appeal to the theory of branching processes to analyse propagation in both scenarios and we discuss the role and the impact of effective participation in stopping misinformation. We show that the propagation reduces drastically provided we assume an active environment, and we obtain theoretical and computational results to measure such a reduction, which in turns depends on the proportion of aware individuals and the number of potential contacts of each individual which is assumed to be random.

\end{abstract}

\section{Introduction}


A sequence of events occurring during the last decade is indicating the beginning of a new phase regarding information transmission at a global scale. In a recent report published by the Council of Europe, Wardle and Derakhshan argue that this new phase involves: information pollution at a global scale; a complex web of motivations for creating, disseminating and consuming these ‘polluted’ messages; a myriad of content types and techniques for amplifying content; innumerable platforms hosting and reproducing this content; and breakneck speeds of communication between trusted peers, see \cite{claire/derak}. 
The debate about how to deal with the exponential growth of false contents is increasing in the scientific literature, see for example \cite{nara,ecker,van} and the references therein. 

For a deeper discussion of information disorder, from a conceptual point of view to a summary of related research, reports and practical initiatives connected to the subject, we refer the reader to \cite{claire/derak,claire}. In \cite{claire/derak} the focus is mainly on two types of information disorder: (1) disinformation, defined as information that is false and deliberately created to harm a person, social group, organization or country; and (2) misinformation, term used to describe information that is false, but not created with the intention of causing harm. In other words, misinformation is a false content shared by a person who does not realize it is false or misleading \cite{claire}. The connection between these two types is at the beginning of the propagation process. Creators of disinformation, independently of their motivation, find spaces for its spreading formed by individuals who do not realize it is false or misleading. Thus we have the beginning of misinformation spreading.

Inspired by the phenomenon of misinformation spreading, we propose two simple mathematical models to illustrate how an effective participation of aware individuals may contribute in stopping, or slowing down, the propagation. Roughly speaking, we consider the spreading of misinformation on a population containing, say, a proportion $1-p$ of people who recognize that the information has been designed to mislead and, as a consequence, they do not propagate it. We call these individuals the aware ones. In other words, we shall assume that any individual will spread the misinformation with probability $p\in(0,1)$. In addition, for the sake of simplicity, we shall assume that any individual will try to pass the information to a random number $\xi$ of individuals. In this case say that propagation occurs in a passive environment, and we shall denote the respective mathematical model, to be formally defined later, as $(p,\xi)$-MPE model. See Figure \ref{fig:model1}.

\begin{figure}[h!]
\begin{center}
\begin{tikzpicture}

\draw[gray,thick] (3,1) -- (3,-5);

\Vertex[x=-2,y=1.5,label={\bf \large a},fontcolor=white, color=black]{A}
\node at (0,0) {\begin{tikzpicture}
\Vertex[color=red, opacity =.7, size=.3]{A} \Vertex[x=2,color=blue,opacity =.3, size=.3]{B} \Vertex[x=1.8,y=-1.8,color=blue,opacity =.3, size=.3]{C} \Vertex[y=-2,color=blue,opacity =.3, size=.3]{D}
\Edge(A)(B)
\Edge(A)(C)
\Edge(A)(D)
\node at (-0.35,-2) {$p$};
\node at (2,0.35) {$p$};
\node at (2,-2.3) {$1-p$};
\end{tikzpicture}};

\Vertex[x=-2,y=-2.5,label={\bf \large b},fontcolor=white, color=black]{B}
\node at (0,-4) {\begin{tikzpicture}
\Vertex[color=red, opacity =.7, size=.3]{A} \Vertex[x=2, color=red, opacity =.7, size=.3]{B} \Vertex[x=1.8,y=-1.8, color=blue, opacity =.3, size=.3]{C} \Vertex[y=-2, color=red, opacity =.7, size=.3]{D}
\Edge(A)(B)
\Edge[label=$+$](A)(C)
\Edge(A)(D)
\end{tikzpicture}};
 
\Vertex[x=10,y=1.5,label={\bf \large c},fontcolor=white, color=black]{C}
\node [rotate=20] at (7,-1.5) {\begin{tikzpicture}
\Vertex[color=red, opacity =.7, size=.3]{A} \Vertex[x=2, color=red, opacity =.7, size=.3]{B} \Vertex[x=1.8,y=-1.8, color=blue, opacity =.3, size=.3]{C} \Vertex[y=-2, color=red, opacity =.7, size=.3]{D}
\Vertex[x=1,y=-3.5, color=blue,opacity =.3, size=.3]{E} \Vertex[y=-4,color=red, opacity =.8, size=.3]{F} \Vertex[x=-1, y=-3.5,color=blue,opacity =.3, size=.3]{G}
\Vertex[x=3.5,y=1, color=blue,opacity =.3, size=.3]{H} \Vertex[x=4,color=red, opacity =.8, size=.3]{I} \Vertex[x=3.5, y=-1,color=red,opacity =.8, size=.3]{J}
\Edge(A)(B)
\Edge[label=$+$](A)(C)
\Edge(A)(D)
\Edge[label=$+$](D)(E)
\Edge[label=$+$](D)(G)
\Edge(D)(F)
\Edge[label=$+$](B)(H)
\Edge(B)(I)
\Edge(B)(J)
\end{tikzpicture}};
\end{tikzpicture}
\caption{Illustration of the $(p,\xi)$-MPE model: propagation of misinformation in a passive environment. Red particles represent propagators and blue particles represent exposed individuals. Each exposed individual is an aware one with probability $1-p$.  (a) Here one individual exposes the misinformation to three individuals, of which one is an aware one. (b) Only, exposed non-aware individuals try to propagate the misinformation. (c) The procedure is repeated provided there are non-aware individuals exposed to the misinformation.}\label{fig:model1}
\end{center}
\end{figure}
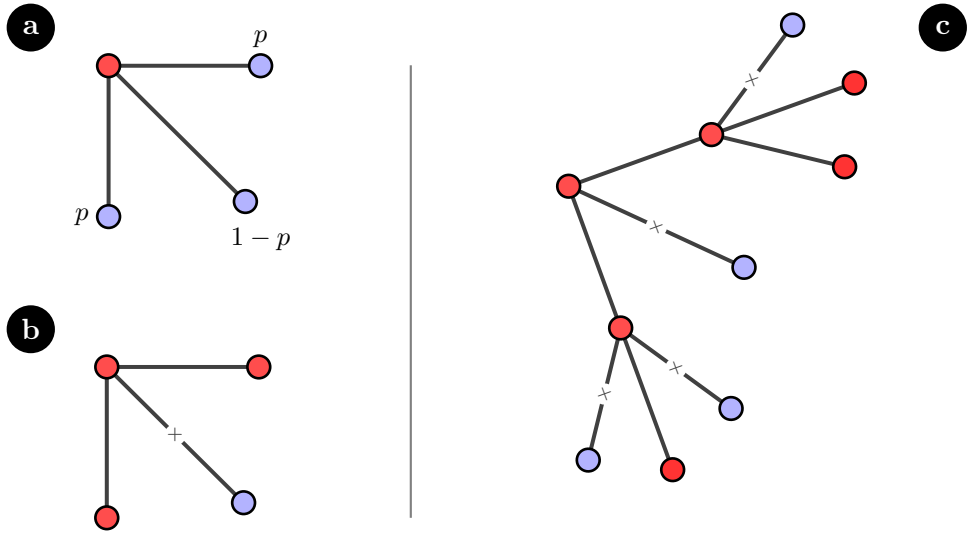

For the second model we assume that if one of an aware individual is exposed to the misinformation then he/she will not spread it and, in addition, such an individual will stop the propagation to other individuals from the individual who contacted him/her. In this case we say that propagation occurs in an active environment, and we shall denote the associated mathematical model as $(p,\xi)$-MAE model. See Figure \ref{fig:modelo2}. For both models, we say that there is extinction of the misinformation if at some time the processes stops, other case we say that the misinformation propagates through the population. 

\begin{figure}
\begin{center}
\begin{tikzpicture}

\Vertex[x=-2,y=1.5,label={\bf \large a},fontcolor=white, color=black]{A}
\node at (0,0) {\begin{tikzpicture}
\Vertex[color=red, opacity =.7, size=.3]{A} \Vertex[x=2,color=red,opacity =.7, size=.3]{B} \Vertex[x=1.8,y=-1.8,color=gray,opacity =.2, size=.3]{C} \Vertex[y=-2,color=gray,opacity =.2, size=.3]{D}
\Edge(A)(B)
\Edge[style={dashed}](A)(C)
\Edge[style={dashed}](A)(D)
\node at (2,0.35) {$p$};
\node at (2,-2.3) {\color{white}$1-p$};
\end{tikzpicture}};

\Vertex[x=3,y=1.5,label={\bf \large b},fontcolor=white, color=black]{B}
\node at (5,0) {\begin{tikzpicture}
\Vertex[color=red, opacity =.7, size=.3]{A} \Vertex[x=2, color=red, opacity =.7, size=.3]{B} \Vertex[x=1.8,y=-1.8, color=blue, opacity =.3, size=.3]{C} \Vertex[y=-2, color=gray,opacity =.2, size=.3]{D}
\Edge(A)(B)
\Edge(A)(C)
\Edge[style={dashed}](A)(D)
\node at (2,-2.3) {$1-p$};
\draw [->,line width=0.4mm,blue!30] (1.8,-1.5) to [out=90, in=330]  (0.4,-0.15);
\end{tikzpicture}};

\Vertex[x=8,y=1.5,label={\bf \large c},fontcolor=white, color=black]{B}
\node at (10,0) {\begin{tikzpicture}
\Vertex[color=blue, opacity =.3, size=.3]{A} \Vertex[x=2, color=red, opacity =.7, size=.3]{B} \Vertex[x=1.8,y=-1.8, color=blue, opacity =.3, size=.3]{C} \Vertex[y=-2, color=gray,opacity =.2, size=.3]{D}
\Edge(A)(B)
\Edge(A)(C)
\Edge[label=$\times$](A)(D)
\node at (2,-2.3) {\color{white}$1-p$};
\end{tikzpicture}};

\draw[gray,thick] (0,-2) -- (10,-2);

\Vertex[x=1,y=-3,label={\bf \large d},fontcolor=white, color=black]{D}
\node [rotate=330] at (5,-5.5) {\begin{tikzpicture}
\Vertex[color=blue, opacity =.3, size=.3]{A} \Vertex[x=2, color=red, opacity =.7, size=.3]{B} \Vertex[x=1.8,y=-1.8, color=blue, opacity =.3, size=.3]{C} \Vertex[y=-2, color=gray,opacity =.2, size=.3]{D}
\Vertex[x=3.5,y=1, color=gray,opacity =.2, size=.3]{H} \Vertex[x=4,color=gray,opacity =.2, size=.3]{I} \Vertex[x=3.5, y=-1,color=blue,opacity =.3, size=.3]{J}

\draw [->,line width=0.4mm,blue!30] (3.3,-1.1) to [out=180, in=280]  (2.1,-0.25);

\Edge(A)(B)
\Edge(A)(C)
\Edge[label=$\times$](A)(D)
\Edge[style={dashed}](B)(H)
\Edge[style={dashed}](B)(I)
\Edge(B)(J)
\end{tikzpicture}};

\draw[gray,thick] (0,-8) -- (10,-8);

\Vertex[x=1,y=-9,label={\bf \large e},fontcolor=white, color=black]{E}
\node [rotate=330] at (5,-11.5) {\begin{tikzpicture}
\Vertex[color=blue, opacity =.3, size=.3]{A} \Vertex[x=2, color=blue, opacity =.3, size=.3]{B} \Vertex[x=1.8,y=-1.8, color=blue, opacity =.3, size=.3]{C} \Vertex[y=-2, color=gray,opacity =.2, size=.3]{D}
\Vertex[x=3.5,y=1, color=gray,opacity =.2, size=.3]{H} \Vertex[x=4,color=gray,opacity =.2, size=.3]{I} \Vertex[x=3.5, y=-1,color=blue,opacity =.3, size=.3]{J}
\Edge(A)(B)
\Edge(A)(C)
\Edge[label=$\times$](A)(D)
\Edge[label=$+$](B)(H)
\Edge[label=$\times$](B)(I)
\Edge(B)(J)
\end{tikzpicture}};

\end{tikzpicture}
\caption{Illustration of the $(p,\xi)$-MAE model: propagation of misinformation in an active environment. Red particles represent propagators, blue particles represent exposed individuals, and gray particles represent potential exposed individuals. Each exposed individual is an aware one with probability $p$.  (a) In a first attempt, one individual exposes the misinformation to another individual, who is an non-aware one. (b) At a second attempt, the first individual exposes the misinformation to an aware one. In this case, the aware individual has an active role and stops the propagation, generating the scenario represented in (c). (d)-(e) The procedure is repeated provided there are non-aware individuals exposed to the misinformation.}\label{fig:modelo2}
\end{center}
\end{figure}
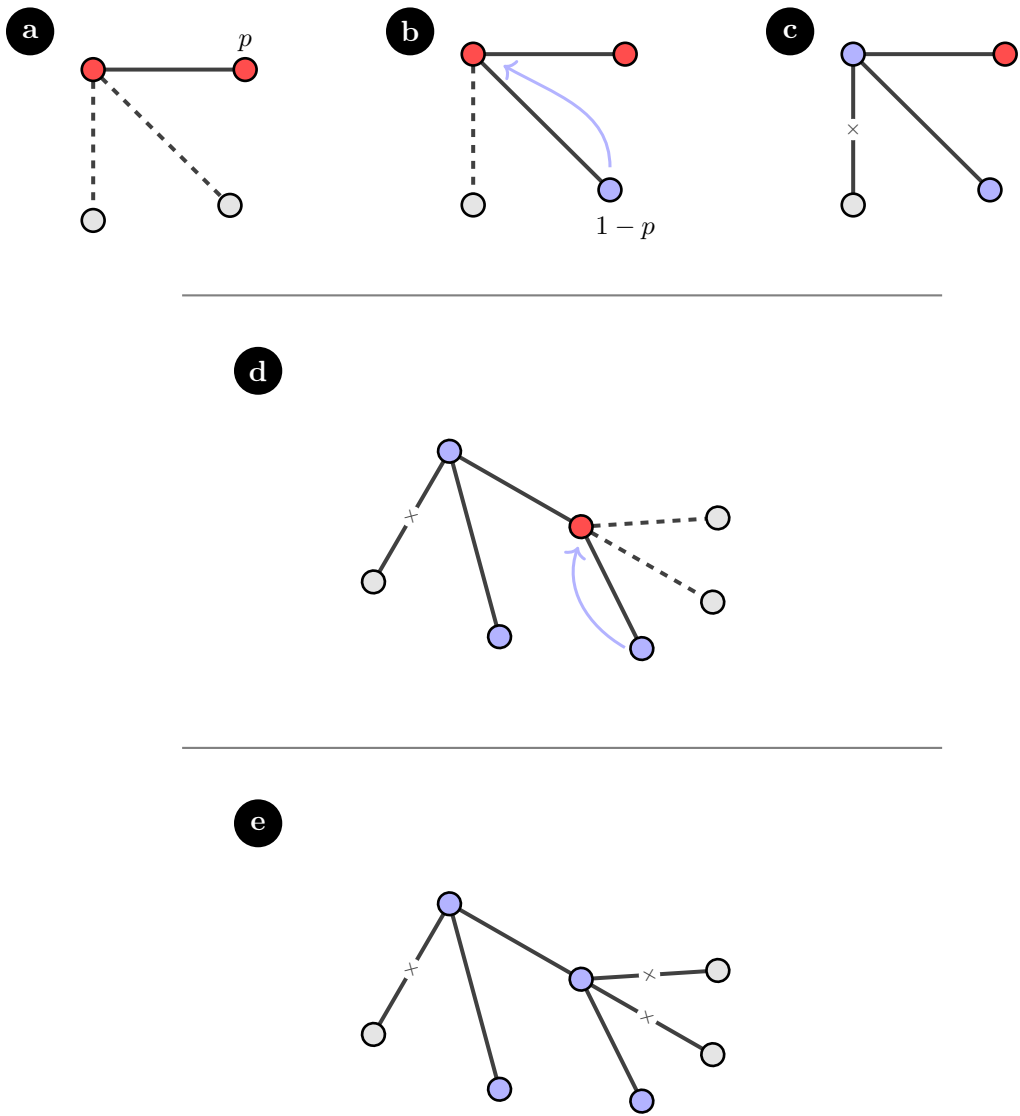

Formally, our models are branching processes which in turn are probabilistic models that already were shown to be accurate for the description of the early stages of an epidemic outbreak or information propagation, see \cite{andersson,draief,gleeson}. Indeed, such stochastic processes serve as approximations for propagation on networks once the actual structure of the network including correlations and loops is neglected. Thus, we appeal to the theory of branching processes to compare both scenarios represented by the models and so understand the impact of effective participation in stopping misinformation.

We notice that by considering $p=1$ in both of our models we recover the propagation model considered by \cite{gleeson}. In such a work the authors applied branching processes to model Twitter-based information cascades and they shown that existing and new theoretical results for branching processes are in agreement with many statistical characteristics of the empirical information cascades, reinforcing the applicability of branching processes for the description of information propagation. In our work we show that a random perturbation in the model, namely by assuming a proportion of aware individuals, allows us to analyse and compare different scenarios of propagation through a suitable application of well-known results coming from the Theory of Branching Processes. In particular, our $(p,\xi)$-MPE model may be seen as a discrete-time susceptible-infected-recovered epidemic model, SIR model for short, on random trees. As far as we know, one of the first works dealing with a discrete-time SIR model, from a probabilistic point of view, is \cite{draief2} to model the spreading of computer viruses, worms and other self-replicating malware. The authors appeal to probabilistic tools to study the behavior of the average number of infected individuals in some specific network models of practical interest, like Erd\"os–Rényi and power law random graphs. Their work is restricted to finite graphs, and it was complemented later by \cite{farkhondeh} who found a simple lower bound on the expected number of ever infected vertices using a breath-first search algorithm. The continuous-time version of the SIR model is well-known and we refer the reader to \cite{andersson,draief,draief2} and the references therein for a review of techniques and results for dealing with such models. We also refer the reader to \cite{moreno-PhysA2007,arruda2,arruda3} for a review of other models for information transmission on finite populations, and to \cite{agliari,AVPHom,AVPRan} for a deeper discussion of the application of branching processes to study the Maki-Thompson rumor model on small-world graphs and (random) trees. We point out that we appeal to branching processes and suitable properties of probability generating functions to study and compare the $(p,\xi)$-MPE model and the $(p,\xi)$-MAE model, which in turn can inspire new research related to propagation models on finite graphs.

We divide our work into two parts. First we perform a similar analysis than the one done by \cite{gleeson} to reinforce the accuracy of using branching processes to represent information transmission in social networks. Like in \cite{gleeson} we analyse data from Twitter, different from the one analysed in such a work. This is included in Section 2 after a review of the concept of branching process. The second part of our work is devoted to the comparison of two possible scenarios in a process of misinformation spreading. We represent the spreading of misinformation in both a passive and an active environment through suitably defined branching processes. Section 3 is devoted to the formulation of such processes and the statement of phase transition results related to the propagation or extinction of the misinformation. In Section 3 we also formulate rigorous results related to the tree cascade height and the tree cascade size for both models, which are random quantities useful to obtain a taste of how many individuals become a vector in the information process. A discussion comparing the two represented scenarios is summarized in Section 4, and we include all the proofs of theorems in Section 5.

\section{Branching processes and information spreading}

\subsection{What is a branching process?} The Theory of Branching Processes is one of the many branches of Probability Theory which serves as an example of the fact that simple models, inspired in observable phenomena, can evolve to become a strong branch of Applied Mathematics useful to the solution of relevant problems for applied sciences. A branching process can be interpreted as a mathematical (probabilistic) model of the following phenomenon. Imagine that at time $0$ there is a particle, and that such particle gives birth to new particles at time $n=1$. The number of new particles follows the law of a discrete random variable $X$. Then, each particle born at time $n\in \mathbb{N}$, will give birth to new particles, at time $n+1$, according to a random variable independent and identically distributed ($i.i.d.$) to $X$. We call a branching process, with offspring distribution given by $X$, the stochastic process $(Z_n)_{n\geq 0}$ where $Z_n$ denotes the number of particles born at time $n$, for $n\geq 0$. To avoid trivial situations we assume that $\mathbb{P}(X=i)=p_i$, for $i\in \mathbb{N}\cup \{0\}$, with $p_0 >0$, $p_0 + p_1 <1$, and $m:=\mathbb{E}(X)<\infty$. In other words, the stochastic process is such that $Z_0$ is a discrete random variable (for simplicity we take $Z_0 = 1$), and for any $n\in \mathbb{N}$
$$Z_{n+1}:=\sum_{i=1}^{Z_n}X_i,$$
where the $X_i$'s are $i.i.d.$ to the random variable $X$. The beginning of such processes does not belong to mathematicians, but to demographers or biologists, which were interested in the survival or extinction of noble families. One can identify individuals with particles so the particles born in a given time $n$ of the process correspond to the $n$th-generation of the family. An interesting historical review can be found in \cite{jagers}. Since the first interest was in the ``extinction'' or not of the process, one can define this as the event
$$\mathcal{E}:=\bigcup_{n=1}^{\infty}\{Z_n =0\}.$$

It is well known that $\mathbb{P}(\mathcal{E})=1$ if, and only if, $m\leq 1$. Moreover, the extinction probability can be localized as the smallest root of $G_{X}(s)=s$, where $G_X$ is the generating probability function of $X$, see \cite{harris}. We refer the reader to \cite{BranchingProcesses,harris,schinazi} for a deeper discussion of branching processes, its generalizations, and some applications. In particular, it is worth pointing out that the connection between epidemic-like processes and branching processes is well-known, since lasts describe accurately the early stages of an epidemic outbreak, see for example \cite[Section 3.3]{andersson} or \cite{draief}. The link is done in a natural way, infected individuals are represented by particles and each born particle from a given particle in the branching process represents an individual who was infected by the corresponding infected one. This in turns, given the similarities between the transmission of an infection and the spreading of an information \cite{daley_nature}, implies that some tools of branching processes can be applied to the representation of information spreading on a population, see \cite{gleeson,AVPHom,AVPRan}.

\subsection{The underlying branching process of Twitter-induced information cascades}

By following the same approach than the one proposed by \cite{gleeson} we shall reinforce that branching processes may be suitable for the representation of the first stages of Twitter-induced information cascades. We consider the Twitter data sets constructed and analyzed by \cite{jing}, namely Twitter15 and Twitter16, which respectively contains 1,490 and 818 propagation trees. Each data set contains, in the form of tree structures, the spread source tweets and their propagation threads, including retweets. Moreover, since the purpose of such a work was to address the problem of identifying rumors based on their propagation structure, the data sets are classified between false rumors, true rumors and non-rumors. We refer the reader to \cite{jing}, and the references therein, for a deeper discussion about the subject. Here, we shall identify the cascade structures of the different spreading processes in order to associate a suitable defined branching process. We shall consider rooted trees, where we assume that the root is the tweet starting a thread. Then, we represent a retweet of the data set with a vertex of the tree. See Figure \ref{fig:twitter} for an illustration of such identification.

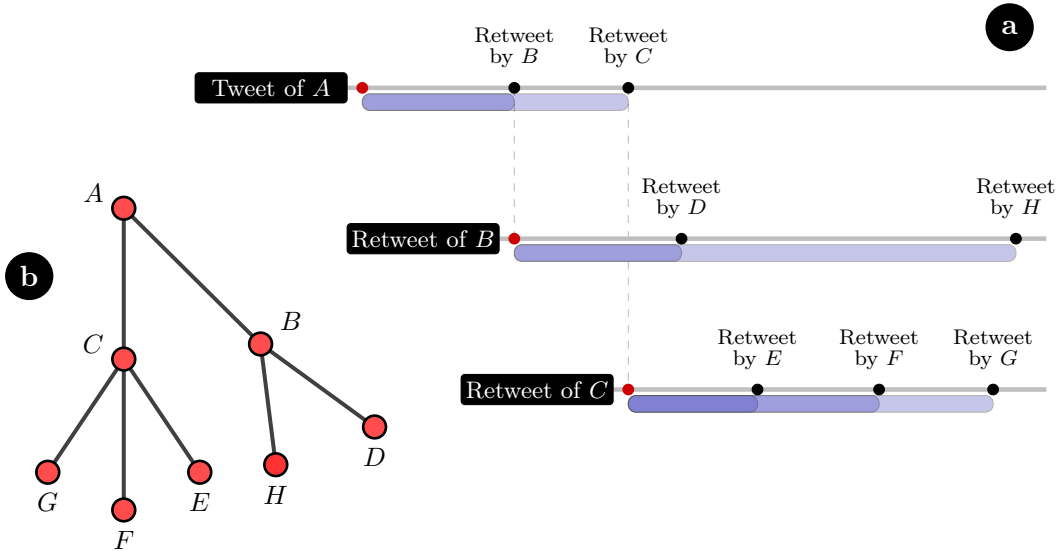
\begin{figure}[h!]
\begin{center}
\begin{tikzpicture}

\Vertex[x=5,y=2.5,label={\bf \large a},fontcolor=white, color=black]{A}

\Vertex[x=-7.9,y=-0.8,label={\bf \large b},fontcolor=white, color=black]{B}

\node at (0,0) {
\begin{tikzpicture}

\draw[gray!50!white, dashed] (2,0) -- (2,-2);
\draw[gray!50!white, dashed] (3.5,0) -- (3.5,-2);
\draw[gray!50!white, dashed] (3.5,-2.3) -- (3.5,-4);

\draw[gray!50!white, line width=0.6mm] (-0.2,0) -- (9,0);
\draw [fill=black,rounded corners=.05cm] (-2.2,-0.2) rectangle (-0.2,0.2);
\draw [fill=blue!50!gray, opacity=0.3,rounded corners=.1cm] (0,-0.3) rectangle (3.5,-0.08);
\draw [fill=blue!50!gray, opacity=0.3,rounded corners=.1cm] (0,-0.3) rectangle (2,-0.08);
\node at (-1.2,0) {\color{white} \small Tweet of $A$};
\filldraw [red!80!black] (0,0) circle (2pt);
\filldraw [black] (2,0) circle (2pt);
\node at (2,0.7) {\footnotesize Retweet};
\node at (2,0.4) {\footnotesize by $B$};
\filldraw [black] (3.5,0) circle (2pt);
\node at (3.5,0.7) {\footnotesize Retweet};
\node at (3.5,0.4) {\footnotesize by $C$};


\draw[gray!50!white, line width=0.6mm] (1.8,-2) -- (9,-2);
\draw [fill=black,rounded corners=.05cm] (-0.2,-2.2) rectangle (1.8,-1.8);
\draw [fill=blue!50!gray, opacity=0.3,rounded corners=.1cm] (2,-2.3) rectangle (8.6,-2.08);
\draw [fill=blue!50!gray, opacity=0.3,rounded corners=.1cm] (2,-2.3) rectangle (4.2,-2.08);
\node at (0.8,-2) {\color{white} \small Retweet of $B$};
\filldraw [red!80!black] (2,-2) circle (2pt);
\filldraw [black] (4.2,-2) circle (2pt);
\node at (4.2,-1.3) {\footnotesize Retweet};
\node at (4.2,-1.6) {\footnotesize by $D$};
\filldraw [black] (8.6,-2) circle (2pt);
\node at (8.6,-1.3) {\footnotesize Retweet};
\node at (8.6,-1.6) {\footnotesize by $H$};


\draw[gray!50!white, line width=0.6mm] (3.3,-4) -- (9,-4);
\draw [fill=black,rounded corners=.05cm] (1.3,-4.2) rectangle (3.3,-3.8);
\node at (2.3,-4) {\color{white} \small Retweet of $C$};
\filldraw [red!80!black] (3.5,-4) circle (2pt);
\draw [fill=blue!50!gray, opacity=0.3,rounded corners=.1cm] 
(3.5,-4.3) rectangle (8.3,-4.08);
\draw [fill=blue!50!gray, opacity=0.3,rounded corners=.1cm] 
(3.5,-4.3) rectangle (6.8,-4.08);
\draw [fill=blue!50!gray, opacity=0.3,rounded corners=.1cm] (3.5,-4.3) rectangle (5.2,-4.08);
\filldraw [black] (5.2,-4) circle (2pt);
\node at (5.2,-3.3) {\footnotesize Retweet};
\node at (5.2,-3.6) {\footnotesize by $E$};
\filldraw [black] (6.8,-4) circle (2pt);
\node at (6.8,-3.3) {\footnotesize Retweet};
\node at (6.8,-3.6) {\footnotesize by $F$};
\filldraw [black] (8.3,-4) circle (2pt);
\node at (8.3,-3.3) {\footnotesize Retweet};
\node at (8.3,-3.6) {\footnotesize by $G$};
\end{tikzpicture}};

\node at (-5.5,-2) {\begin{tikzpicture}
\Vertex[color=red, opacity =.7, size=.3]{A} 
\Vertex[x=1.8,y=-1.8, color=red, opacity =.7, size=.3]{B}
\Vertex[x=3.3,y=-2.9, color=red, opacity =.7, size=.3]{D}
\Vertex[y=-2, color=red, opacity =.7, size=.3]{C}
\Vertex[x=1,y=-3.5, color=red,opacity =.7, size=.3]{E} \Vertex[y=-4,color=red, opacity =.7, size=.3]{F} 
\Vertex[x=-1, y=-3.5,color=red,opacity =.7, size=.3]{G}
\Vertex[x=2,y=-3.4, color=red,opacity =.8, size=.3]{H} \Edge(A)(B)
\Edge(A)(C)
\Edge(B)(D)
\Edge(B)(H)
\Edge(C)(E)
\Edge(C)(G)
\Edge(C)(F)
\node at (-0.4,0.2) {$A$};
\node at (2.2,-1.5) {$B$};
\node at (-0.4,-1.8) {$C$};
\node at (3.3,-3.3) {$D$};
\node at (1,-3.9) {$E$};
\node at (0,-4.4) {$F$};
\node at (-1,-3.9) {$G$};
\node at (2,-3.8) {$H$};
\end{tikzpicture}};
\end{tikzpicture}
\caption{Example of a realization of the process that starts with a Tweet of individual $A$. (a) In this case, individuals $B$ and $C$ retweet the initial message, and their retweets are retweeted by individuals $D$ and $H$ on one side, and by individuals $E$, $F$ and $G$ on the other side. (b) The corresponding rooted tree is generating by identifying individuals by vertices. Who publish the first tweet is the root, who is connected with those individuals who retweeted his/her message, who in turns are also connected with those who retweeted them, and so on. Note that although the times between retweets are not homogeneous, this does not matter for constructing the corresponding tree.}\label{fig:twitter}
\end{center}
\end{figure}

Thus defined such a trees can be interpreted as the realization of a branching process whose offspring distribution is obtained directly from the data. Like in \cite{gleeson} each of the data sets that we consider (false, true and non-rumors) can be represented by an ensemble of $N$ trees where vertices represent particles of an underlying branching process and levels of the tree represent its generations. We let $z_{k,n}$ for the number of vertices at the $nth$-level in the $kth$-tree, for $n\in \mathbb{N}$ and $k\in\{1,\ldots,N\}$, and we let
\begin{equation}\label{eq:gen}
z_n:=\sum_{k=1}^{N}z_{k,n}.
\end{equation}

Since our data sets are formed by small trees (few vertices) we consider each ensemble as a whole so $z_n$ in \eqref{eq:gen} denotes the total number of particles at the $nth$-generation. Now, the dependence of $z_n$ on $n$, for each ensemble of the data sets, is showed in Figures \ref{fig:2015} and \ref{fig:2016}, where we are using log-linear scales for the data sets Twitter15 and Twitter16, respectively. Like in \cite{gleeson} we illustrate the effective branching number for the data sets, which is just an estimation of the branching number of the resulting tree computed by doing $z_{n+1}/z_n$. See Figure \ref{fig:branching}.

\begin{figure}
    \centering
    \subfigure[]{\includegraphics[scale=0.4]{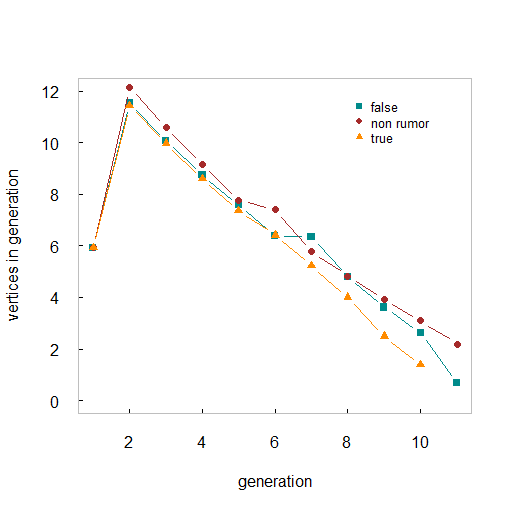}}\subfigure[]{\includegraphics[scale=0.4]{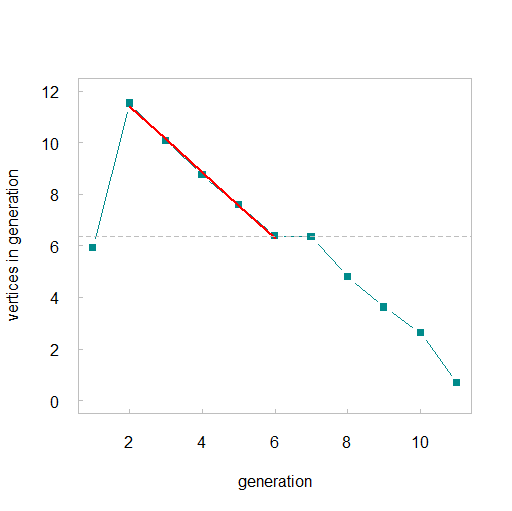}}
    \subfigure[]{\includegraphics[scale=0.4]{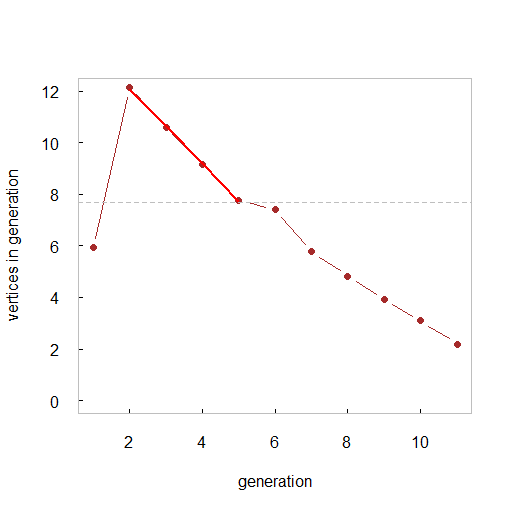}}\subfigure[]{\includegraphics[scale=0.4]{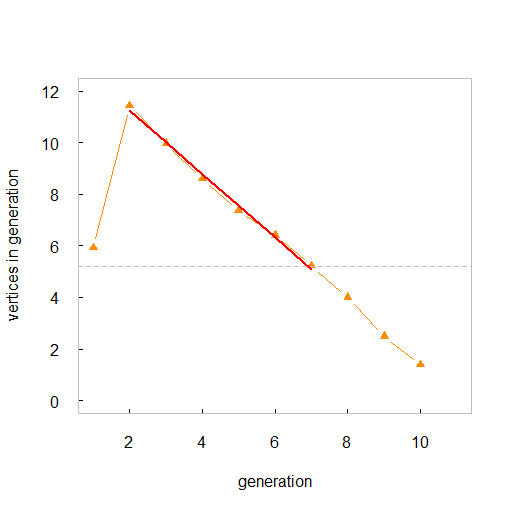}}
    \caption{Number of vertices by generation in the data set Twitter15. (a) Comparative results for false, true and non-rumors. (b)-(d) At the early stages is observed an approximate exponential growth of vertices, seen as linearity using log-linear scales, which is consistent with the growth of a branching process.}
    \label{fig:2015}
\end{figure}

\begin{figure}
    \centering
    \subfigure[]{\includegraphics[scale=0.4]{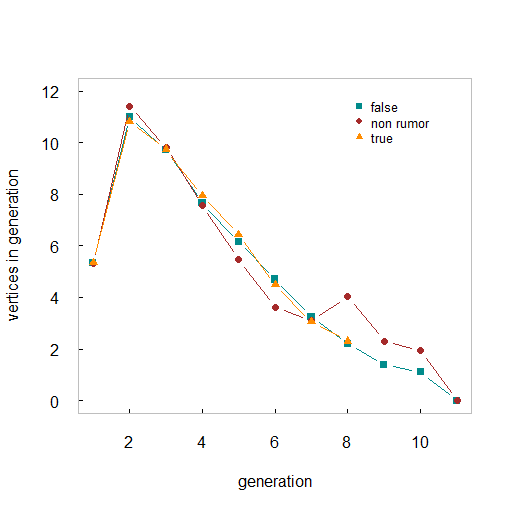}}\subfigure[]{\includegraphics[scale=0.4]{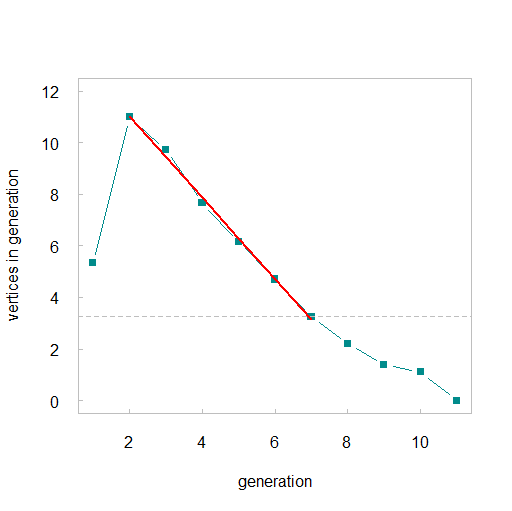}}
    \subfigure[]{\includegraphics[scale=0.4]{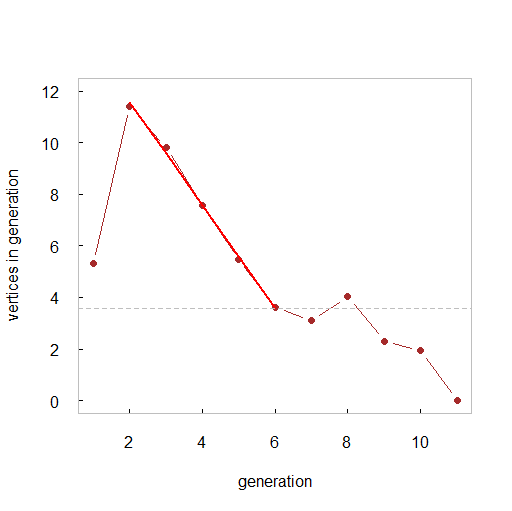}}\subfigure[]{\includegraphics[scale=0.4]{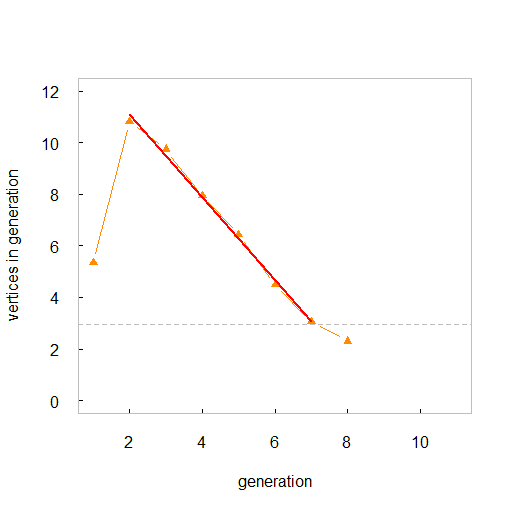}}
    \caption{Number of vertices by generation in the data set Twitter16. (a) Comparative results for false, true and non-rumors. (b)-(d) At the early stages is observed an approximate exponential growth of vertices, seen as linearity using log-linear scales, which is consistent with the growth of a branching process. }
    \label{fig:2016}
\end{figure}

\begin{figure}
    \centering
    \subfigure[Twitter15]{\includegraphics[scale=0.5]{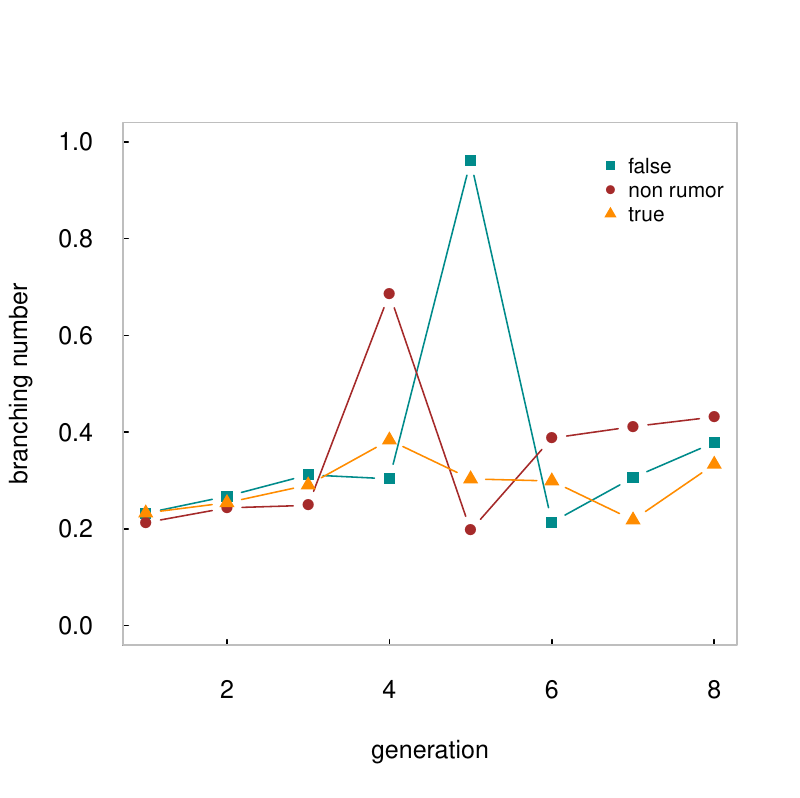}}
        \subfigure[Twitter16]{\includegraphics[scale=0.5]{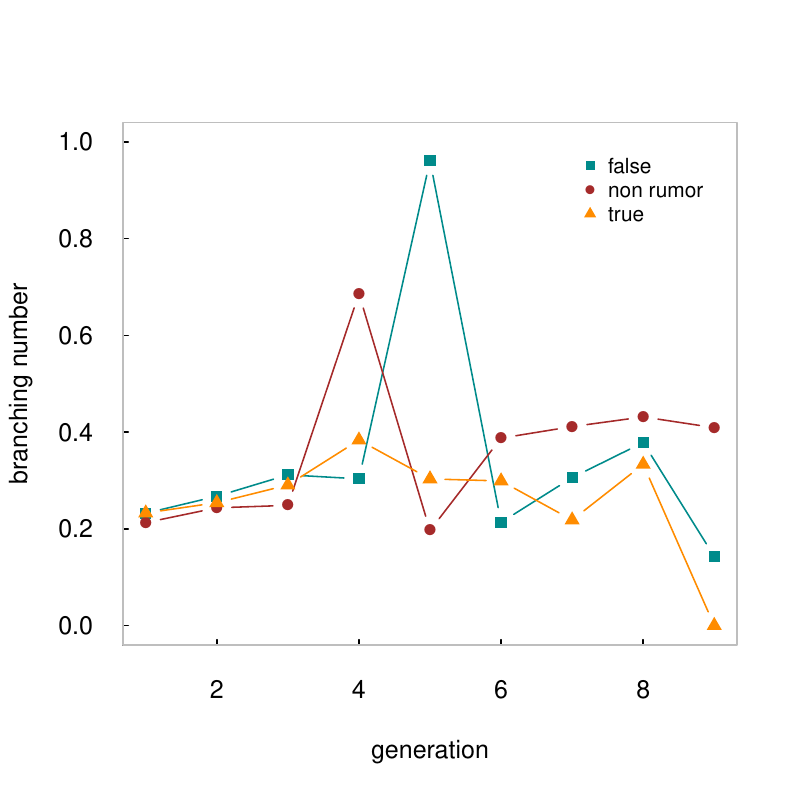}}
    \caption{Effective branching number for the data sets.}
    \label{fig:branching}
\end{figure}

\section{Two misinformation stochastic models}

We appeal to the Theory of Branching Processes to formulate two simple models for misinformation spreading. The purpose of the section is to apply our theoretical models to compare different scenarios of propagation. Roughly speaking, we consider the spreading of misinformation on a population containing, say, a proportion $1-p$ of people who recognize that the information has been designed to mislead and, as a consequence, they do not propagate it. We call these individuals the aware ones. In other words, we shall assume that any individual will spread the misinformation with probability $p\in(0,1)$. In addition, for the sake of simplicity, we shall assume that any individual will try to pass the information to a random number $\xi$ of individuals. In what follows, $\xi$ denotes a random variable taking values in the non-negative integers and such that $\mathbb{E}(\xi)<\infty$.

\subsection{A misinformation in a passive environment stochastic model} 
Suppose that if one of an aware individual is exposed to the misinformation then he/she will not spread it. However, we assume that such an individual will not try to stop the propagation to other individuals from the individual who contacted him/her. This is the reason why we call it of a passive environment. To represent this situation with a probabilistic model we assume that the misinformation is propagating as the following branching process. The process starts from one individual by assuming that at time zero only he/she is carrying the misinformation. At time $n=1$ we assume that a random number $\xi$ of individuals is exposed to the misinformation and choose carrying it with probability $p\in(0,1)$, other case he/she stays neutral in the process. At each discrete time of the process we repeat this procedure with those $\xi$ individuals which are contacted by someone who started carrying the misinformation in the previous time. Indeed, for any new individual we are using random copies of the random variable $\xi$; that is, we are using a sequence of independent and identically distributed random variables with the same law than $\xi$. We say that there is extinction of the misinformation if at some time the processes stops, other case we say that there is propagation of the misinformation. See Figure \ref{fig:model1}. 

Thus defined the sequence of the number of spreaders at the different times of the process form a branching process with offspring distribution given by a binomial distribution with parameters $\xi$ and $p$. We call this stochastic process the $(p,\xi)$ misinformation in a passive environment stochastic model, or the $(p,\xi)$-MPE model for short. The two key results in this section are Theorem \ref{thm:model1}, which shows a phase-transition result in the sense that misinformation propagates when the expectation value of $\xi$ is greater that $1/p$, and Theorem \ref{thm:model1alcance-size}, which specifies the size of the misinformation outbreak.

\begin{thm}\label{thm:model1}
Consider the $(p,\xi)$-MPE model and let $\theta_1(p,\xi)$ be the probability of misinformation propagation. Then $\theta_1(p,\xi)>0$ if, and only if, $p>1/\mathbb{E}(\xi)$. Moreover, $\theta_1(p,\xi)=1-\alpha$ where $\alpha$ is the smallest root of 
$  G_{\xi}(sp+1-p) = s,$ and $G_{\xi}$ is the generating probability function of $\xi$.
\end{thm}

Theorem \ref{thm:model1} gains in interest if we realize that it allow us to exhibit the phase-diagram of extinction $vs$ propagation of misinformation, according to the values of $p$ and $\mathbb{E}(\xi)$. See Figure \ref{fig:phasediag-mod1}.

	\begin{figure}[!h]
	\centering
\includegraphics[scale=0.8]{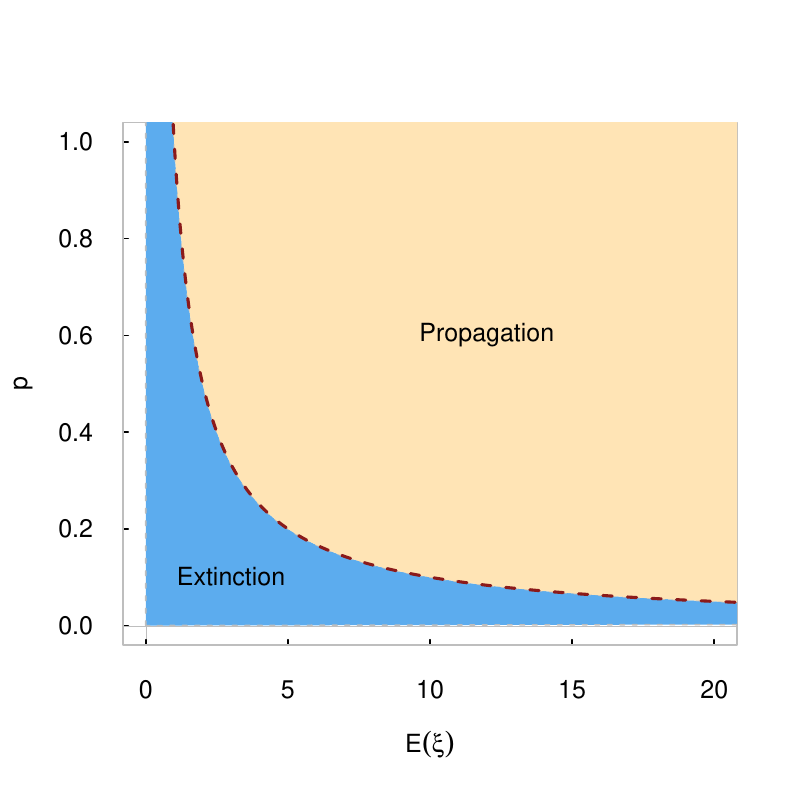}
\caption{Phase-diagram for the $(p,\xi)$-MPE model. For $p<1/\mathbb{E}(\xi)$ there is extinction almost surely of the misinformation. For $p>1/\mathbb{E}(\xi)$ the model exhibits propagation of the misinformation with positive probability.}\label{fig:phasediag-mod1}
\end{figure}

In order to illustrate the generality of Theorem \ref{thm:model1}, we discuss the result for two special cases of $\xi$. First we assume a constant number $k$ of exposed individuals at any time. That is, we consider $\mathbb{P}(\xi=k)=1$ for some $k\in \mathbb{N}$.

\begin{cor} \label{cor:model1tfc}
Let $k\in\mathbb{N}$ be fixed and consider the $(p,\xi)$-MPE model with $\mathbb{P}(\xi = k) = 1$. Then $\theta_1(p,k)>0$ if, and only if $p > 1/k$. Moreover, $\theta_1(p,k)=1-\alpha$ where $\alpha$ is the smallest root of $(sp+1-p)^k = s$.
\end{cor}

It is well-known that a power-law probability distribution arises naturally for the description of social, and many other types of observable phenomena, see \cite{barabasi, castellano, muchnik} and the references therein. The second special case in our discussion is when $\xi$ follows a Zipf distribution with parameter $\gamma$, $\gamma>1$, which is considered a reasonable distribution for fitting the degree sequence of real networks; see \cite{duarte} and references therein. That is, we consider 
$$\mathbb{P}(\xi = j) = \frac{ j^{-\gamma}}{\zeta(\gamma)},$$ 
for $j\in\mathbb{N}$, and a given $\gamma>1$, where 
$$\zeta(\gamma)=\sum_{i=1}^{\infty}\frac{1}{i^{\gamma}}$$
is the Riemann zeta function. We adopt the notation $\xi\sim Zipf(\gamma)$ for such a distribution. 

\begin{cor} \label{cor:model1tfzeta}
Consider the $(p,\xi)$-MPE model, with $\xi\sim Zipf(\gamma)$, $\gamma>1$. Then $\theta_1(p,\gamma)>0$ if, and only if $p > \zeta(\gamma)/\zeta(\gamma-1)$. Moreover, $\theta_1(p,\gamma)=1-\alpha$ where $\alpha$ is the smallest root of $ Li_{\gamma}(sp+1-p) = s\zeta(\gamma)$, where 
$$ Li_{a}(x) := \displaystyle \sum_{j=1}^{\infty} = \frac{x^j}{j^a} , |x| < 1$$ 
is the polylogarithm function also known as the Jonquière's function.
\end{cor}




 Two random quantities useful to obtain a taste of how many individuals become a vector in the information process are the tree cascade height and the tree cascade size. The first one is the highest level of the tree reached by the information, while the second is the total number of individuals reached by the information. Our representation through branching processes allow us to obtain rigorous results for these quantities. In what follows we focus our attention in the extinction case, which is the case when the process stops almost surely.

\smallskip
 \begin{thm} \label{thm:model1alcance}
(Tree cascade height) Consider the $(p,\xi)$-MPE model, with $p < 1/\mathbb{E}(\xi)$, and let $\mathcal{T}_h$ be the tree cascade height. Suppose that for all $s \in [0,1]$
\begin{equation}
     G_{\xi}^{\prime \prime}(sp+1-p) + G_{\xi}^\prime(sp+1-p) \geq \frac{\mathbb{E}(\xi^2)}{\mathbb{E}(\xi)}G_{\xi}^\prime(sp+1-p),
\end{equation}

\noindent
and let

\begin{equation}
u := u(p,\xi)= 1 + \frac{[1- G_{\xi}(1-p)][1-p\mathbb{E}(\xi)]}{p\mathbb{E}(\xi) + G_{\xi}(1-p) -1 },\,\,\, \ell := \ell(p,\xi)= 1 + \frac{2\mathbb{E}(\xi)[1 -p \mathbb{E}(\xi)]}{p \mathbb{E}(\xi(\xi-1))}.
\end{equation}

Then\[ \frac{1-[ p\mathbb{E}(\xi)]^{n+1} }{ 1- \ell^{-1}[p\mathbb{E}(\xi)]^{n+1}} \leq \mathbb{P}( \mathcal{T}_h \leq n) \leq  \frac{1-[ p\mathbb{E}(\xi)]^{n+1} }{ 1- u^{-1}[p\mathbb{E}(\xi)]^{n+1}}
\]

and

\[ (u-1) \sum_{n=0}^{\infty}  \left [ \frac{[p\mathbb{E}(\xi)]^{n+1}}{u - [p\mathbb{E}(\xi)]^{n+1}}  \right ] \leq \mathbb{E}(\mathcal{T}_h) \leq (\ell-1) \sum_{n=0}^{\infty}  \left [ \frac{[p\mathbb{E}(\xi)]^{n+1}}{\ell - [p\mathbb{E}(\xi)]^{n+1}}  \right ]. 
\]
\end{thm}

\bigskip
Theorem \ref{thm:model1alcance} is technical but help us to obtain some conclusions provided the law of $\xi$ is well-known. 

\bigskip
\begin{cor} \label{cor:model1alcance}
Let $k\in\mathbb{N}$ be fixed and consider the $(p,\xi)$-MPE model with $\mathbb{P}(\xi = k) = 1$, and $p < 1/k$. Let $m:=m(k,p)=kp$, and 
\[ u := \frac{m(1-p)^k}{m + (1-p)^k -1} \textrm { , } \ell := \frac{2-m-p}{m-p}.
\]
Then
\[ \frac{1-m^{n+1} }{ 1- \ell^{-1}m^{n+1}} \leq \mathbb{P}( \mathcal{T}_h \leq n) \leq  \frac{1-m^{n+1} }{ 1- u^{-1}m^{n+1}}
\]

\noindent 
and

\[ (u-1) \sum_{n=0}^{\infty}  \left [ \frac{m^{n+1}}{u - m^{n+1}}  \right ] \leq \mathbb{E}(\mathcal{T}_h) \leq (\ell-1) \sum_{n=0}^{\infty}  \left [ \frac{m^{n+1}}{\ell - m^{n+1}}  \right ]. 
\]

\end{cor}

\begin{figure}[h!]
    \centering
    \subfigure[]{\includegraphics[scale=0.3]{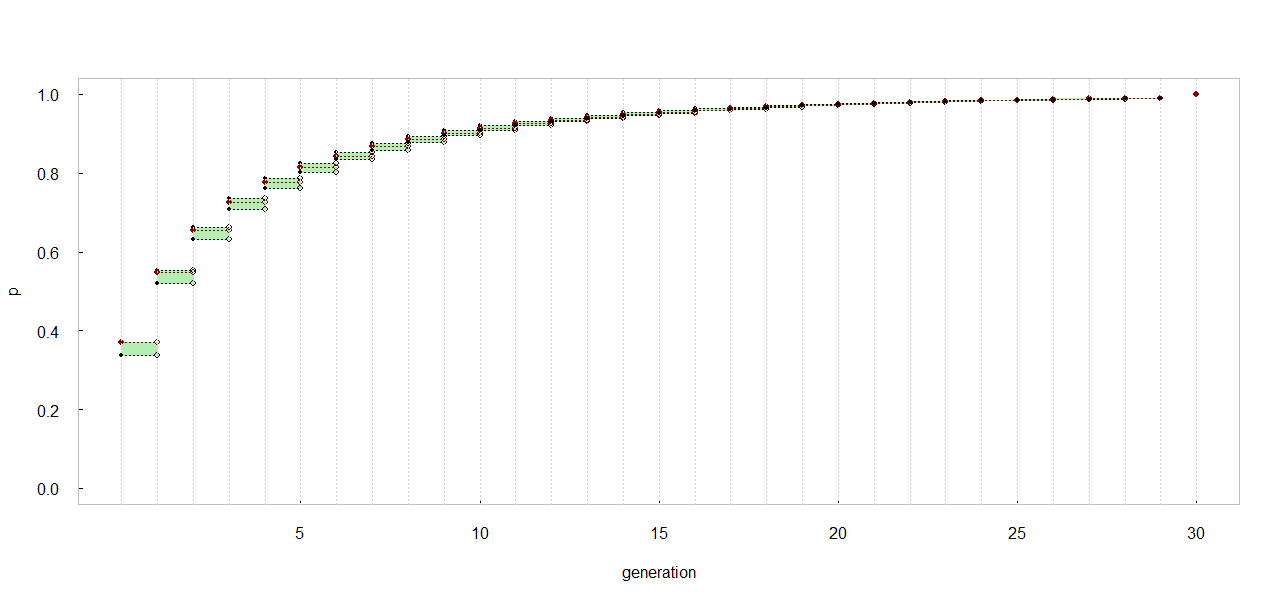}}
        \subfigure[]{\includegraphics[scale=0.4]{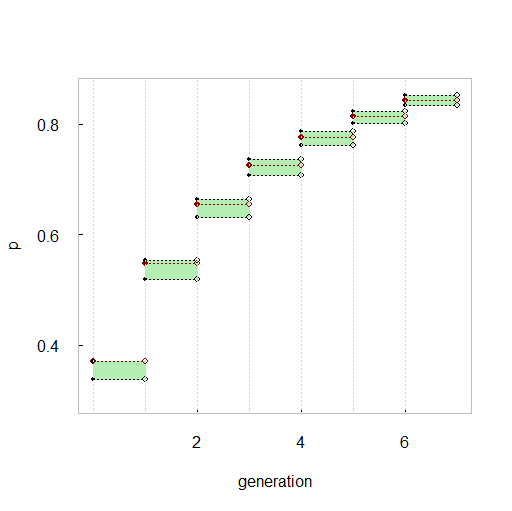}}\subfigure[]{\includegraphics[scale=0.4]{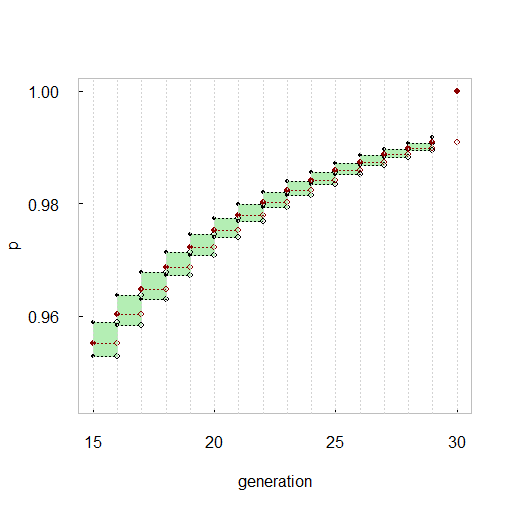}}
    \caption{Tree cascade height distribution, $\mathbb{P}( \mathcal{T}_h \leq n)$, for the $(p,\xi)$-MPE model, for $p=0.5$ and $\mathbb{P}(\xi = 5)=1$. Black dashed lines represent the theoretical bounds obtained in Corollary \ref{cor:model1alcance}(i). Red dashed lines represent the approximation obtained from $8\times 10^6$ simulations of the process.}
    \label{fig:enter-label}
\end{figure}

\bigskip
\begin{cor} \label{cor:model1alcancezeta}
 Consider the $(p,\xi)$-MPE model, with $\xi\sim Zipf(\gamma)$, $\gamma>1$, and $p < \zeta(\gamma)/\zeta(\gamma-1)$. Suppose that for all $s \in [0,1]$
 \[ p \geq \frac{1}{1-s}\left [ 1 - \frac{\zeta(\gamma-1)[Li_{\gamma -2}(sp+1-p) - Li_{\gamma -1}(sp+1-p)]}{Li_{\gamma -1}(sp+1-p)[\zeta(\gamma-2) - \zeta(\gamma-1)]}   \right ],
 \]

 \noindent
 and let
\[ u := 1 +\frac{[\zeta(\gamma) - Li_{\gamma}(1-p)][\zeta(\gamma) -p\zeta(\gamma-1)]}{p\zeta(\gamma-1) +Li_{\gamma}(1-p) -\zeta(\gamma)}, \,\,\, \ell := 1+ \frac{2\zeta(\gamma-1)[\zeta(\gamma)-p\zeta(\gamma-1)]}{\zeta(\gamma)[\zeta(\gamma-2) - \zeta(\gamma-1)]}.
\]

Then\[ \frac{1-\left(\frac{p\zeta(\gamma-1)}{\zeta(\gamma)} \right )^{n+1} }{ 1- \ell^{-1}\left(\frac{p\zeta(\gamma-1)}{\zeta(\gamma)} \right )^{n+1}} \leq \mathbb{P}( \mathcal{T}_{h} \leq n) \leq  \frac{1-\left(\frac{p\zeta(\gamma-1)}{\zeta(\gamma)} \right )^{n+1} }{ 1- u^{-1}\left(\frac{p\zeta(\gamma-1)}{\zeta(\gamma)} \right )^{n+1}}
\]

and
\[ (u-1) \sum_{n=0}^{\infty}  \left [ \frac{\left(\frac{p\zeta(\gamma-1)}{\zeta(\gamma)} \right )^{n+1}}{u - \left(\frac{p\zeta(\gamma-1)}{\zeta(\gamma)} \right )^{n+1}}  \right ] \leq \mathbb{E}(\mathcal{T}_h) \leq (\ell-1) \sum_{n=0}^{\infty}  \left [ \frac{\left(\frac{p\zeta(\gamma-1)}{\zeta(\gamma)} \right )^{n+1}}{\ell - \left(\frac{p\zeta(\gamma-1)}{\zeta(\gamma)} \right )^{n+1}}  \right ] 
.\]
\end{cor}

\bigskip
As above, the Theory of Branching Processes allow us to obtain some information about the tree cascade size. Indeed, that we call tree cascade size here is the total progeny of the respective branching process. 

\bigskip
 \begin{thm} \label{thm:model1alcance-size}
(Tree cascade size) Consider the $(p,\xi)$-MPE model, with $p < 1/\mathbb{E}(\xi)$, and let $\mathcal{T}_s$ be the tree cascade size. Then
\[ \mathbb{P}(\mathcal{T}_s = j) = \frac{1}{j!}G^{(j-1)}(0), \textrm { where } G(s) = (G_{\xi}(sp+1-p))^j.
\]
Moreover,
\[ \mathbb{E}(\mathcal{T}_s) = \frac{1}{1-p\mathbb{E}(\xi)}.
\] 

\end{thm}

\bigskip
\begin{cor} \label{cor:model1alcance-size}
Let $k\in\mathbb{N}$ be fixed, and consider the $(p,\xi)$-MPE model with $\mathbb{P}(\xi = k) = 1$, and $p < 1/k$. Let $m:=kp$. Then 
\[ \mathbb{P}(\mathcal{T}_s = j) = \frac{1}{j} \binom{kj}{j-1}p^{j-1}(1-p)^{(k-1)j+1}, j\in\mathbb{N}
\]
and
\[ \mathbb{E}(\mathcal{T}_s) = \frac{1}{1-m}.
\] 
\end{cor}

\bigskip
\begin{cor} \label{cor:model1alcancezeta-size}
Consider the $(p,\xi)$-MPE model, with $\xi\sim Zipf(\gamma)$, $\gamma>1$, and $p < \zeta(\gamma)/\zeta(\gamma-1)$. Then
\[ \mathbb{P}(\mathcal{T}_s = j) = \frac{1}{j!}G^{(j-1)}(0), \textrm { where } G(s) = \left ( \frac{Li_{\gamma}(sp+1-p)}{\zeta(\gamma)} \right)^j
\]
and
\[ \mathbb{E}(\mathcal{T}_s) = \frac{\zeta(\gamma)}{\zeta(\gamma) - p\zeta(\gamma-1)}.
\] 
\end{cor}

\bigskip
\subsection{A misinformation in an active environment stochastic model}       

Now we assume that if one of an aware individual is exposed to the misinformation then he/she will not spread it and, in addition, such an individual will stop the propagation to other individuals from the individual who contacted him/her. We call it of an active environment. As before, we represent this situation by assuming that the misinformation is propagating as a branching process. 
The process starts from the one individual by assuming that at time zero only he/she is carrying the misinformation, while the other individuals are ignorant. At time $n=1$ we assume that at most a random number $\xi$ of individuals will be exposed to the misinformation. This will happens once at a time, and each exposed individual will choose carrying the information with probability $p$, other case he/she stays neutral in the propagation and will successful stop the propagation. At each discrete time of the process we repeat this procedure. See Figure \ref{fig:modelo2}. Thus defined the sequence of the number of spreaders at the different times of the process form a branching process with offspring distribution given by a $\xi$-truncated geometric distribution of parameter $1-p$, and we call this stochastic process the $(p,\xi)$ misinformation in an active environment stochastic model, or the $(p,\xi)$-MAE model for short. Here, we use the form of geometric distribution used for modeling the number of failures until the first success, where we assume that success (the contacted person is an aware individual) occurs with probability $1-p$ while failure (the contacted person is a non-aware individual) occurs with probability $p$. That is, the offspring distribution of the related branching process is given by a random variable $X=X(p,\xi)$, where $X=Y I_{\{Y\leq \xi\}} + \xi I_{\{Y>\xi\}}$, where the random variable $Y$ is such that
$$\mathbb{P}(Y=i)=(1-p)p^i,\text{ for }i\in\mathbb{N}\cup\{0\},$$
and $I_{A}$ denotes the indicator random variable of event $A$. Now the main results in this section are Theorem \ref{thm:model2}, which shows a necessary and sufficient condition on the law of $\xi$ to guarantee a phase-transition result, and Theorem \ref{thm:model2alcance-size}, which specifies the size of the misinformation outbreak for the $(p,\xi)$-MAE model.


\begin{thm} \label{thm:model2}
Consider the $(p,\xi)$-MAE model. If $\theta_2(p, \xi )$ denotes the probability of misinformation propagation, then $\theta_2(p,\xi)>0$ if, and only if, 
\[ \sum_{x=1}^{\infty}xp^x \mathbb{P}(\xi \geq x) > \frac{1-p\mathbb{E}(\xi p^{\xi})}{1-p}.
\]
Moreover, $\theta_2(p,\xi)=1-\alpha$ where $\alpha$ is the smallest root of 
\[  \sum_{x=0}^{\infty}(sp)^x \mathbb{P}(\xi \geq x) = \frac{s-p G_{\xi}(sp)}{1-p}.
\]
\end{thm}

\bigskip
\begin{cor} \label{cor:model2tfc}
Let $k\in\mathbb{N}$ be fixed and consider the $(p,\xi)$-MAE model with $\mathbb{P}(\xi = k) = 1$. Then $\theta_2(p,k)>0$ if, and only if 
$$k>\left(\frac{\log(2p-1)}{\log p}\right)-1.$$
Moreover, $\theta_2(p,\xi)=1-\alpha$ where $\alpha$ is the smallest root of $p^{k+1}s^{k+1} - p^{k+1}s^{k} -ps^2 +s -1 + p = 0$.
\end{cor}

\bigskip
\begin{cor} \label{cor:model2tfzeta}
Consider the $(p,\xi)$-MAE model, with $\xi\sim Zipf(\gamma)$, $\gamma>1$. Let 
\begin{equation}
h(s,p, \gamma) := \sum_{x=1}^{\infty}(sp)^x\zeta(\gamma,x),
\end{equation}
where $\zeta(\gamma,x) = \displaystyle \sum_{i=0}^{\infty} \frac{1}{(i+x)^{\gamma}}$ is the Hurwitz zeta Function.
Then $\theta_2(p,\gamma)>0$ if, and only if
\[ \dfrac { \partial h(1,p,\gamma)}{\partial s} > \frac{\zeta(\gamma) -pLi_{\gamma -1}(p)}{1-p}.
\]

Moreover, $\theta_2(p,\gamma)=1-\alpha$ where $\alpha$ is the smallest root of
\[ h(s,p, \gamma)  = \frac{s\zeta(\gamma) -p Li_{\gamma}(sp)}{1-p},
\]
where $ Li_{a}(x), |x| < 1$, is the poly-logarithm function also known as the Jonquière's function.
\end{cor}

In what follows we obtain similar results than the ones obtained in the previous section related to the tree cascade height and the  tree cascade size. The reader will see that this case is a bit more technical than the previous one. Section 4 will be useful to compare both scenarios.    

\bigskip
\begin{thm} \label{thm:model2alcance}
(Tree cascade height) Consider the $(p,\xi)$-MAE model, and let $\mathcal{T}_h$ be the tree cascade height. Let
\[ \mu_G(s) = \frac{(1-p)}{s}\sum_{x=1}^{\infty}x(ps)^x \mathbb{P}(\xi \geq x) + \frac{p}{s}\mathbb{E}(\xi (ps)^{\xi}) , \] 

\[B_G(s) = \frac{(1-p)}{s^2}\sum_{x=2}^{\infty}x(x-1)(ps)^x \mathbb{P}(\xi \geq x) + \frac{p}{s^2}\mathbb{E}(\xi(\xi-1)(ps)^{\xi}).
\]
and write $\mu_G := \mu_G(1)$ and $B_G := B_G(1).$ Suppose that $\mu_GB_G(s) - \mu_G(s)B_G \geq 0,  s \in [0,1]$, and let 
\[ u := \frac{[1-p\mathbb{P}(\xi \neq 0)]\mu_G}{\mu_G - p(\mathbb{P}(\xi \neq 0))} \textrm { and } \ell := \frac{2\mu_G(1-\mu_G) +B_G}{B_G}. 
\] If \[ \sum_{x=1}^{\infty}xp^x \mathbb{P}(\xi \geq x) < \frac{1-p\mathbb{E}(\xi p^{\xi})}{1-p}
\]
then
\[ \frac{1-\mu_G^{n+1} }{ 1- \ell^{-1}\mu_G^{n+1}} \leq \mathbb{P}( \mathcal{T}_h \leq n) \leq  \frac{1-\mu_G^{n+1} }{ 1- u^{-1}\mu_G^{n+1}}
\]
and
\[ (u-1) \sum_{n=0}^{\infty}  \left [ \frac{\mu_G^{n+1}}{u -\mu_G^{n+1}}  \right ] \leq \mathbb{E}(\mathcal{T}_h) \leq (\ell-1) \sum_{n=0}^{\infty}  \left [ \frac{\mu_G^{n+1}}{\ell - \mu_G^{n+1}}  \right ]. 
\]
\end{thm}

\bigskip
\begin{cor}\label{cor:model2alcancek}
Let $k\in\mathbb{N}$ be fixed and consider the $(p,\xi)$-MAE model with $\mathbb{P}(\xi = k) = 1$. 
Let
\[ \mu_G(s) = \frac{1}{s} \left [ (1-p) \sum_{x=1}^{k}x(ps)^x +kp(sp)^k \right]
\]
and
\[  B_G(s) =  \frac{1}{s^2} \left [ (1-p) \sum_{x=2}^{k}x(x-1)(ps)^x +k(k-1)p(sp)^k \right]. 
\]
Write $\mu_G := \mu_G(1)$ and $B_G := B_G(1).$ Suppose that $\mu_GB_G(s) - \mu_G(s)B_G \geq 0,  s \in [0;1]$, and let 
\[ u := \frac{(1-p)(1-p^k)}{p(1-p^{k-1})}  \textrm { and } \ell := \frac{3p^{k+1} -p^{2k+1} +[k(p-1)-p]p^k -p+1}{-p^{k+1} +k(p-1)p^k +p} . 
\]
If  $p^{k+1} - 2p+1 > 0$ then
\[ \frac{1-\left(\frac{p(1-p^k)}{1-p} \right)^{n+1} }{ 1- \ell^{-1}\left(\frac{p(1-p^k)}{1-p} \right)^{n+1}} \leq \mathbb{P}( \mathcal{T}_h \leq n) \leq  \frac{1-\left(\frac{p(1-p^k)}{1-p} \right)^{n+1} }{ 1- u^{-1}\left(\frac{p(1-p^k)}{1-p} \right)^{n+1}}
\]
and
\[ (u-1) \sum_{n=0}^{\infty}  \left [ \frac{\left(\frac{p(1-p^k)}{1-p} \right)^{n+1}}{u - \left(\frac{p(1-p^k)}{1-p} \right)^{n+1}}  \right ] \leq \mathbb{E}(\mathcal{T}_h) \leq (\ell-1) \sum_{n=0}^{\infty}  \left [ \frac{\left(\frac{p(1-p^k)}{1-p} \right)^{n+1}}{\ell - \left(\frac{p(1-p^k)}{1-p} \right)^{n+1}}  \right ] .
\]
\end{cor}

\bigskip
\begin{cor} \label{cor:model2alcancez}
Consider the $(p,\xi)$-MAE model, with $\xi\sim Zipf(\gamma)$, $\gamma>1$. Let
\[ \mu_G(s) = \frac{1}{s \zeta(\gamma)}\left [(1-p)\dfrac { \partial h(s,p, \gamma)}{\partial s} + p Li_{\gamma -1}(ps) \right ]
\]
\[  B_G(s) =  \frac{1}{s^2\zeta(\gamma)} \left [(1-p)\dfrac { \partial^2 h(s,p, \gamma)}{\partial s^2} + p [Li_{\gamma -2}(ps) - Li_{\gamma -1}(ps)] \right ], 
\]
\[ \mu_G =  \frac{1}{\zeta(\gamma)}\left [(1-p)\dfrac { \partial h(1,p, \gamma)}{\partial s} + p Li_{\gamma -1}(p) \right ] 
\]
and
\[ B_G =  \frac{1}{\zeta(\gamma)} \left [(1-p)\dfrac { \partial^2 h(1,p, \gamma)}{\partial s^2} + p [Li_{\gamma -2}(p) - Li_{\gamma -1}(p)] \right ]. 
\] 

Suppose that $\mu_GB_G(s) - \mu_G(s)B_G \geq 0,  s \in [0;1]$, and let
\[ u = \frac{1+p}{\mu_G} \textrm { and } \ell = \frac{2\mu_G(1-\mu_G) +B_G}{B_G}. 
\]
If
\[ \dfrac { \partial h(1,p,\gamma)}{\partial s} < \frac{\zeta(\gamma) -pLi_{\gamma -1}(p)}{1-p},
\]
then
\[ \frac{1-\mu_G^{n+1} }{ 1- \ell^{-1}\mu_G^{n+1}} \leq \mathbb{P}( \mathcal{T}_h \leq n) \leq  \frac{1-\mu_G^{n+1} }{ 1- u^{-1}\mu_G^{n+1}}
\]
and
\[ (u-1) \sum_{n=0}^{\infty}  \left [ \frac{\mu_G^{n+1}}{u - \mu_G^{n+1}}  \right ] \leq \mathbb{E}(\mathcal{T}_h) \leq (\ell-1) \sum_{n=0}^{\infty}  \left [ \frac{\mu_G^{n+1}}{\ell - \mu_G^{n+1}}  \right ]. 
\]

\end{cor}

\bigskip
\begin{thm} \label{thm:model2alcance-size}
(Tree cascade size) Consider the $(p,\xi)$-MAE model, and let $\mathcal{T}_s$ be the tree cascade size. If \[ \sum_{x=1}^{\infty}xp^x \mathbb{P}(\xi \geq x) < \frac{1-p\mathbb{E}(\xi p^{\xi})}{1-p}.
\]
\[ \mathbb{P}(\mathcal{T}_s = k) = \frac{1}{k!}G^{(k-1)}(0), \textrm { where } G(s) = \left ( (1-p)\sum_{x=0}^{\infty} (sp)^x \mathbb{P}(\xi \geq x) + pG_{\xi}(sp)\right )^k,
\]
and
\[ \mathbb{E}(\mathcal{T}_s) = \left (1-p\mathbb{E}(\xi p^{\xi}) -(1-p)\sum_{x=1}^{\infty} xp^x \mathbb{P}(\xi \geq x) \right)^{-1}
\]
\end{thm}

\begin{cor}\label{cor:model2alcancek-size}
Let $k\in\mathbb{N}$ be fixed and consider the $(p,\xi)$-MAE model with $\mathbb{P}(\xi = k) = 1$. If $p^{k+1} - 2p+1 > 0$
\[ \mathbb{P}(\mathcal{T}_s = j) = \frac{1}{j!}G^{(j-1)}(0), \textrm { where } G(s) = \left ( \frac{p(1-s)(ps)^k +1-p}{1-ps}  \right )^j.
\]
and
\[ \mathbb{E}(\mathcal{T}_s) = \frac{1-p}{1-2p+p^{k+1}}.
\] 
\end{cor}

\begin{cor} \label{cor:model2tamanhoz}
Consider the $(p,\xi)$-MAE model, with $\xi\sim Zipf(\gamma)$, $\gamma>1$. If \[ \dfrac { \partial h(1,p,\gamma)}{\partial s} < \frac{\zeta(\gamma) -pLi_{\gamma -1}(p)}{1-p},
\]
then
\[ \mathbb{P}(\mathcal{T}_s = j) = \frac{1}{j!}G^{(j-1)}(0), \textrm { where } G(s) = \left [ \frac{(1-p)h(s,p, \gamma) + pLi_{\gamma}(sp)}{\zeta(\gamma)}\right ]^j.
\]
and
\[ \mathbb{E}(\mathcal{T}_s) = \frac{\zeta(\gamma)}{\zeta(\gamma) - p Li_{\gamma -1}(p) -(1-p) \dfrac { \partial h(1,p, \gamma)}{\partial s}}.
\]
\end{cor}


\section{Misinformation: passive {\it vs} active environments of aware individuals}

We formulated two simple branching processes to represent misinformation propagation through two environments of aware individuals. We called it as passive and active environments, respectively. Our approach is useful to obtain rigorous results from the Theory of Branching Processes, which in turns allows to compare the different scenarios represented by the models. Theorems \ref{thm:model1} and \ref{thm:model2} state a phase transition result related to the propagation with positive probability, or not, of misinformation. For the sake of simplicity let us consider the case $\mathbb{P}(\xi=k)=1$ to compare both models, see Corollaries \ref{cor:model1tfc} and \ref{cor:model2tfc}. In Figure \ref{fig:compa-diag} we show the phase-diagram of the $(p,\xi)$-MPE model (a) in contrast to the phase-diagram of the $(p,\xi)$-MAE model (b). We emphasize that there exists an expressive reduction in the propagation region when we assume an active environment. In particular, differently from what happens in a passive environment, in an active environment there is extinction of misinformation almost surely for any value of $k$ provided the proportion of aware individuals is greater than $1/2$. See Figure \ref{fig:compa-diag}.


	\begin{figure}[h!]
	\centering\subfigure[$(p,\xi)$-MPE model.]{\includegraphics[scale=0.55]{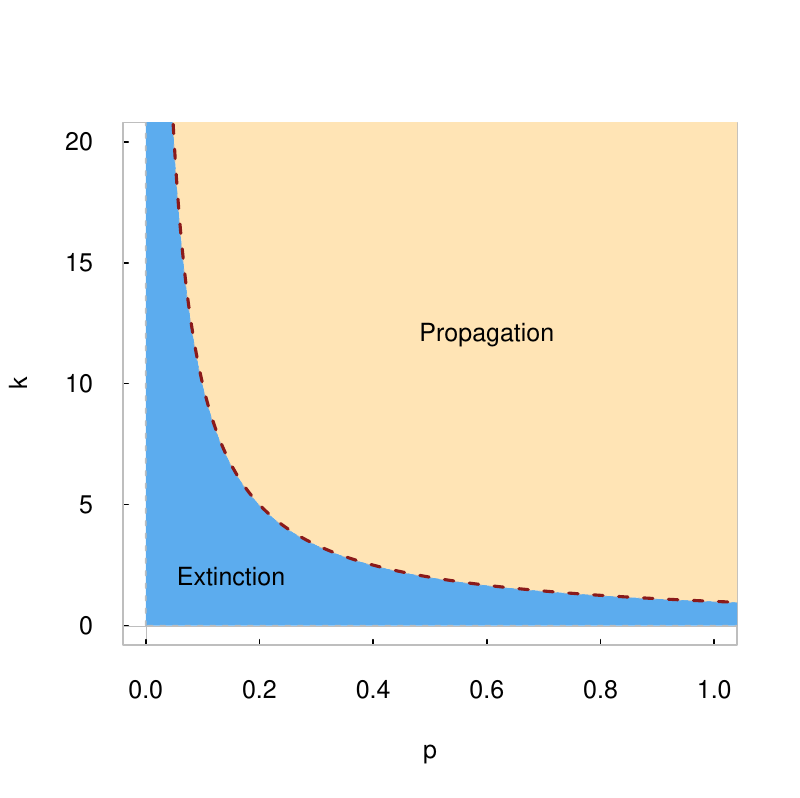}}\subfigure[$(p,\xi)$-MAE model.]{\includegraphics[scale=0.55]{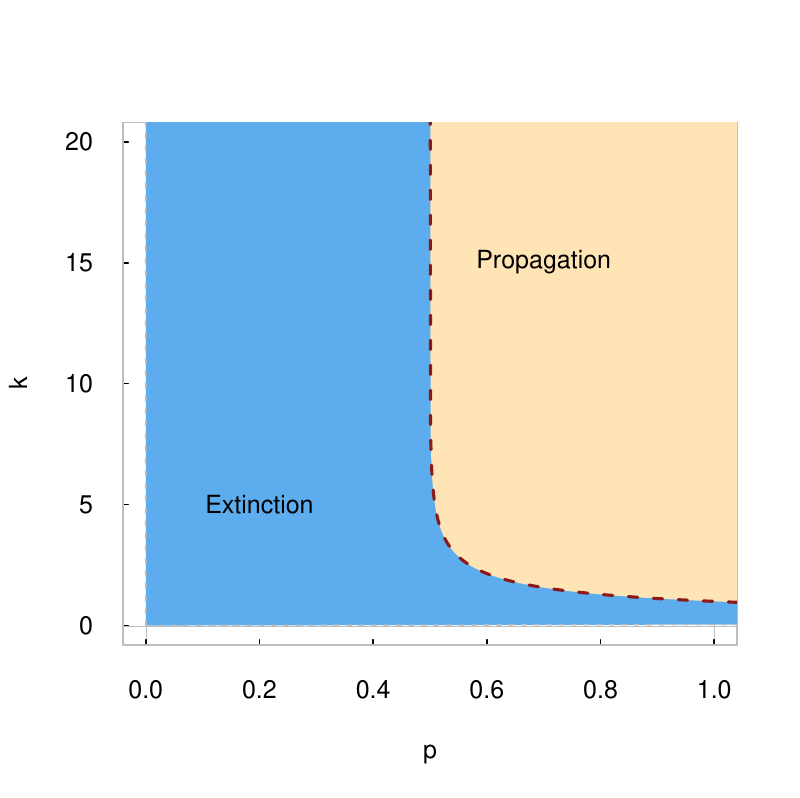}}
\caption{Phase-diagrams for the misinformation stochastic models assuming $\mathbb{P}(\xi=k)=1$. Dashed lines represent the critical $k$ as a function of $p$ for which phase transition occurs. Below such critical function (blue region) there is extinction of misinformation almost surely. Above the critical function (yellow region)  misinformation propagates with positive probability. (a) For the $(p,\xi)$-MPE model Corollary \ref{cor:model1tfc} states that misinformation propagates with positive probability if, and only if, $k>1/p$. Note that $p=0$ is an asymptote for the critical function. (b) For the $(p,\xi)$-MAE model Corollary \ref{cor:model2tfc} states that misinformation propagates with positive probability if, and only if, $k>\left(\log(2p-1)/\log p\right)-1.$ In this case $p=0.5$ is an asymptote for the critical function. In particular, differently from what happens in the $(p,\xi)$-MPE model, if $p\leq 1/2$ then there is extinction of misinformation almost surely for any value of $k$.}\label{fig:compa-diag}
\end{figure}

Another measure for the range of propagation that we studied in this work is the cascade size of each propagation process. We state information about the law of this random variable in Theorems \ref{thm:model1alcance-size} and \ref{thm:model2alcance-size}. Moreover, as a consequence we obtain another way to catch the impact of an active environment over a passive one. We call it as the {\it stopping effect}. 

\begin{thm} \label{thm:comp}
Let $\mathbb{E}_{p}(\mathcal{T}_s)$ and $\mathbb{E}_{a}(\mathcal{T}_s)$ the expected value for the cascade size $\mathcal{T}_s$ for the $(p,\xi)$-MAE and the $(p,\xi)$-MPE model, respectively. Suppose that $p < 1/\mathbb{E}(\xi)$. Then,
\[ \mathbb{E}_{p}(\mathcal{T}_s) - \mathbb{E}_{a}(\mathcal{T}_s) = \frac{p[\mathbb{E}(\xi) - \mathbb{E}(\xi p^{\xi}) ] -(1-p) \sum_{x=1}^{\infty} xp^x \mathbb{P}(\xi \geq x)}{[1-p\mathbb{E}(\xi)][1-p \mathbb{E}(\xi p^{\xi}) - (1-p)\sum_{x=1}^{\infty} xp^x \mathbb{P}(\xi \geq x)]}.\]
\end{thm}


Let us consider a fixed $k\in\mathbb{N}$, and let $\mathbb{P}(\xi = k) = 1$. In this case, if $p < 1/k$ then Theorem \ref{thm:comp} implies
\[ \mathbb{E}_p(\mathcal{T}_s) - \mathbb{E}_a(\mathcal{T}_s)  = \frac{p(p^k-kp+k-1)}{(1-kp)(1-2p+p^{k+1})}.\]
See Tables \ref{tab:k2} and \ref{tab:k3} for an illustration of the reduction of the cascade size according to some values of $p$.

\begin{table}
\caption{Stopping effect for the case where $\mathbb{P}(\xi = 2) = 1$.}\label{tab:k2}
\begin{tabular}{ccccc}\toprule 

  $p$ & $\mathbb{E}_p(\mathcal{T}_s)$ & $\mathbb{E}_a(\mathcal{T}_s)$   & $\mathbb{E}_p(\mathcal{T}_s)$  - $\mathbb{E}_a(\mathcal{T}_s)$  &  reduction $\%$ \\
  \midrule
 0.1 & 1.25 & 1.1236 & 0.1264 & 10.11 $\%$   \\
  \midrule
   0.2 & 1.67 & 1.3158 & 0.3509 & 21.05 $\%$   \\
  \midrule
   0.3 & 2.50 & 1.6394 & 0.8607 & 34.42 $\%$   \\
  \midrule
   0.4 & 5.00 & 2.2727 & 2.7273 & 54.55 $\%$   \\
  \midrule
   0.45 & 10.00 & 2.8777 & 7.1223 & 71.22 $\%$   \\
  \midrule
   0.49 & 50.00 & 3.7051 & 46.2949 & 92.59 $\%$   \\
  \midrule
   0.499 & 500.00 & 3.9683 & 496.0320 & 99.21 $\%$   \\
  \midrule
   0.4999 & 5000.00 & 3.9968 & 4996.00 & 99.99 $\%$   \\
  \bottomrule
\end{tabular}\\
\end{table}

\begin{table}
\caption{Stopping effect for the case where $\mathbb{P}(\xi = 3) = 1$.}\label{tab:k3}
\begin{tabular}{ccccc}\toprule

  $p$ & $\mathbb{E}_p(\mathcal{T}_s)$ & $\mathbb{E}_a(\mathcal{T}_s)$   & $\mathbb{E}_p(\mathcal{T}_s)$  - $\mathbb{E}_a(\mathcal{T}_s)$  &  reduction $\%$ \\
\midrule
 0.1 & 1.43 & 1.1249 & 0.3037 & 10.11 $\%$   \\
  \midrule
   0.2 & 2.5 & 1.3298 & 1.1703 & 21.05 $\%$   \\
  \midrule
   0.3 & 10.00 & 1.7153 & 0.2847& 82.25 $\%$   \\
  \midrule
   0.33 & 100.00 & 1.9042 & 98.0958 & 98.10 $\%$   \\
  \midrule
   0.333 & 1000.00 & 1.9261 & 998.0740 & 99.81 $\%$   \\
  \midrule
   0.3333 & 10000.00 & 1.9283 & 9998.0700 & 99.98 $\%$   \\
  \bottomrule
  \end{tabular}\\
\end{table}

\bigskip
To finish the discussion, let us compare the evolution on time for both models. Since we are dealing with branching processes, which are Markov chains, its simulation are relatively simple. In Figure \ref{fig:comparison} we show the evolution of the propagation for the $(p,\xi)$-MPE model (left side) $vs$ the $(p,\xi)$-MAE model (right side), resulting from simulations, assuming $\mathbb{P}(\xi=k)=1$ for different values of $k$, taking different values of $p$, and using log-linear scales for the results.

\begin{figure}[h!]
    \centering
    \subfigure[$k=5$.]{\includegraphics[scale=0.43]{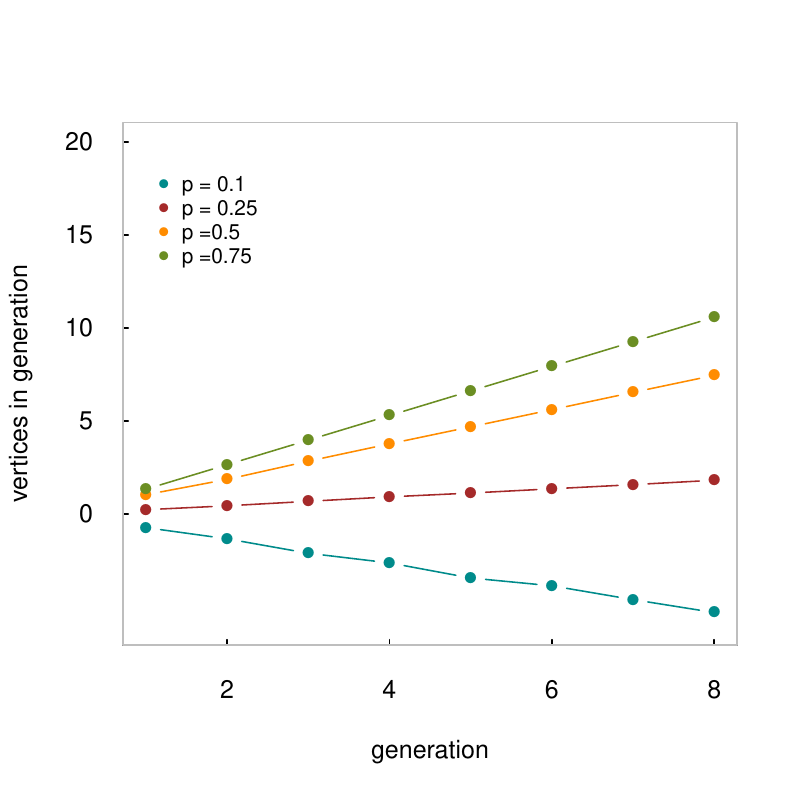}\includegraphics[scale=0.43]{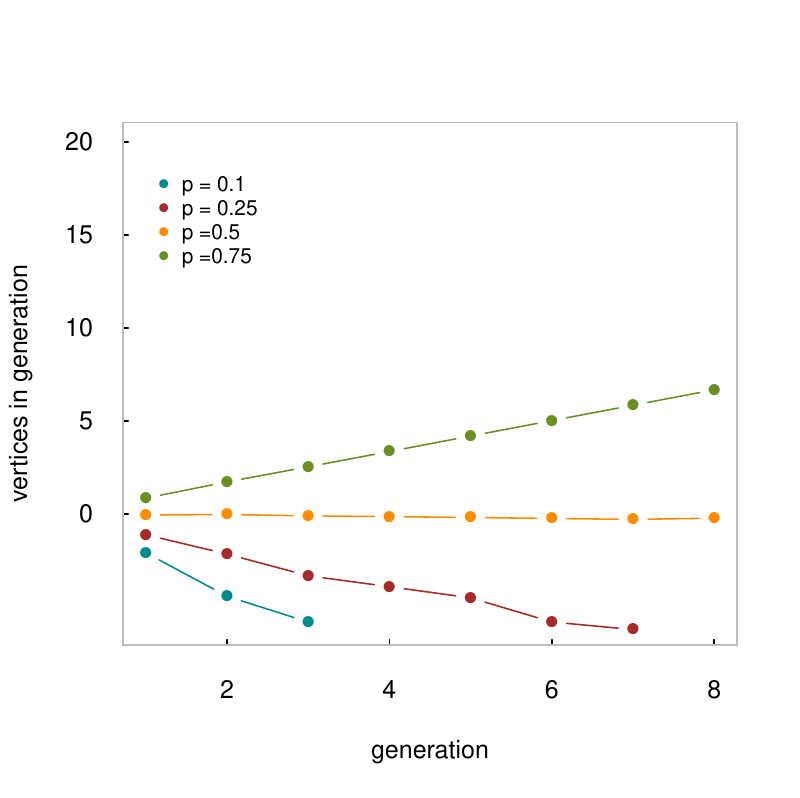}}
    \subfigure[$k=10$.]{\includegraphics[scale=0.43]{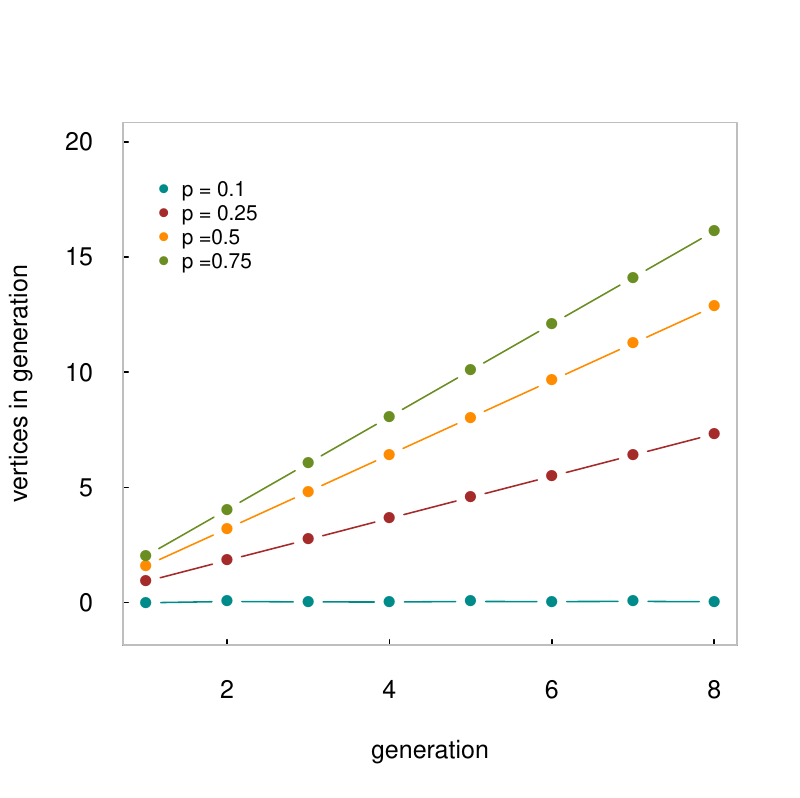}\includegraphics[scale=0.43]{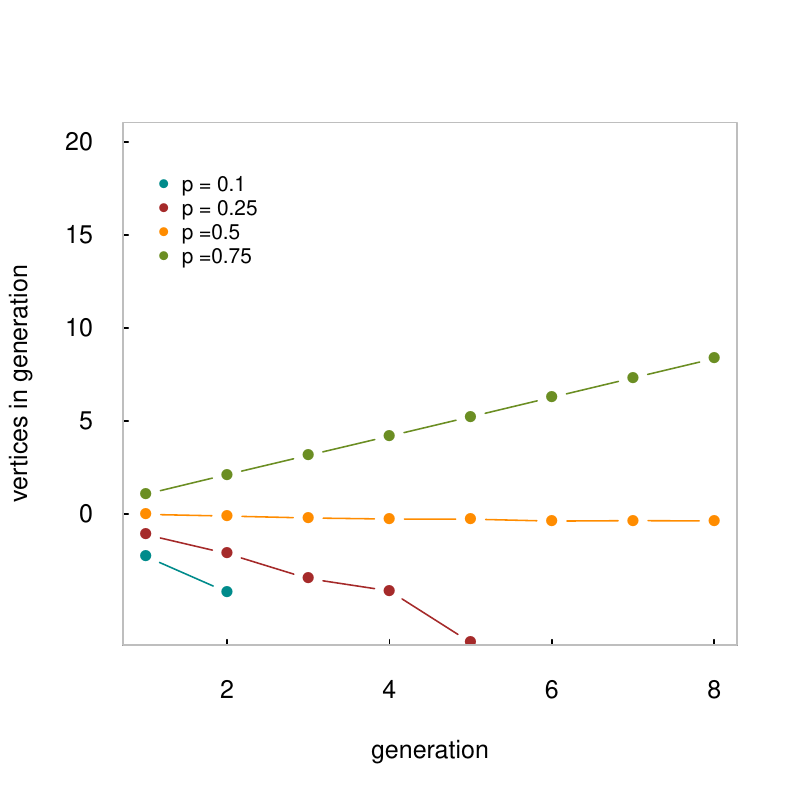}}
    \subfigure[$k=15$.]{\includegraphics[scale=0.43]{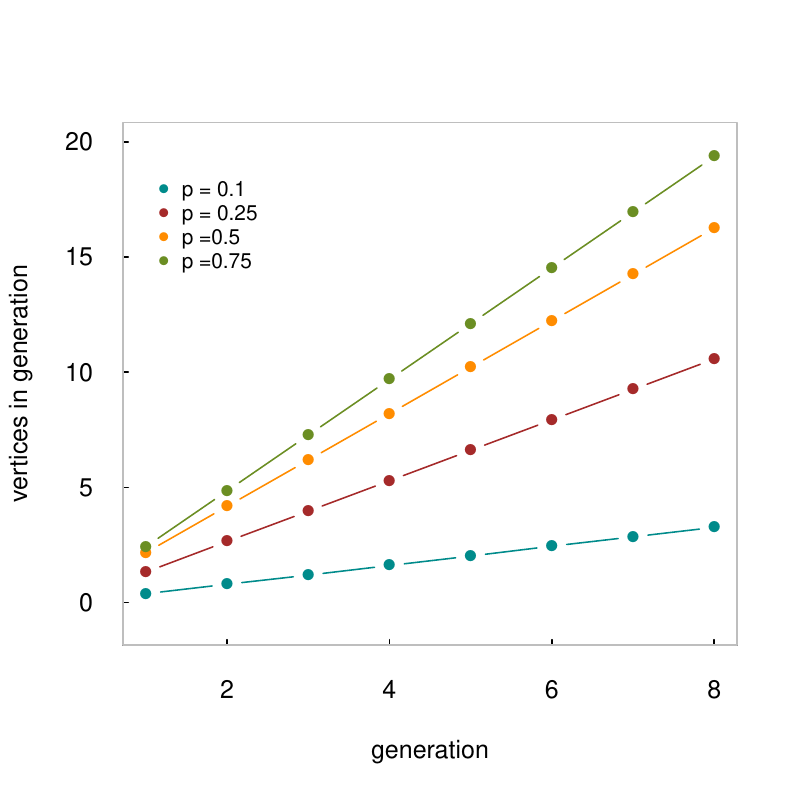}\includegraphics[scale=0.43]{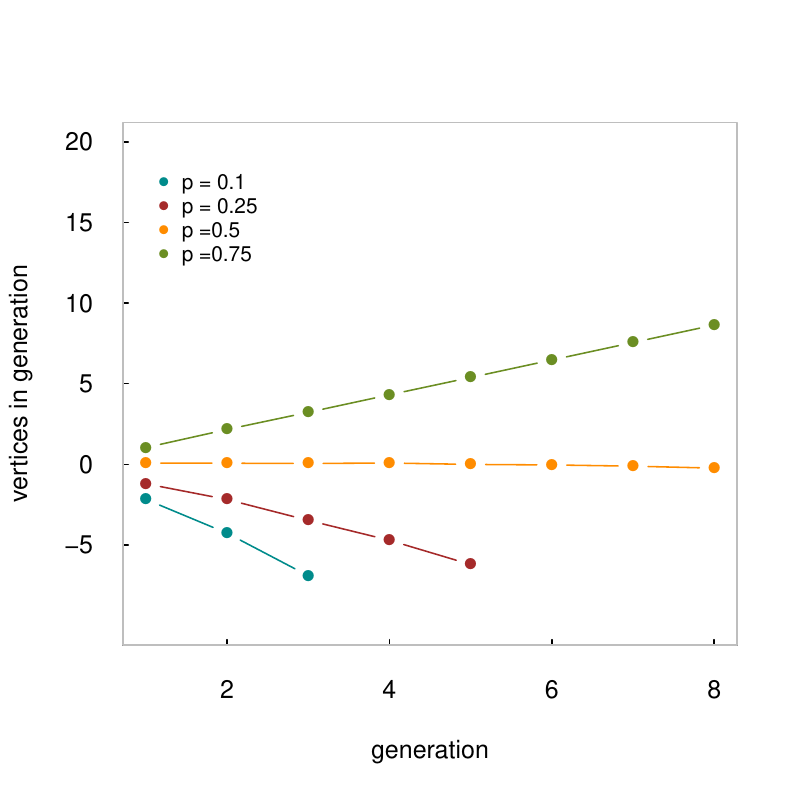}}
    \caption{Evolution of the propagation for the $(p,\xi)$-MPE model (left side) $vs$ the $(p,\xi)$-MAE model (right side). We exhibit the mean values obtained from $8\times 10^6$ simulations of the processes for $p\in\{0.1,0.25,0.5,0.75\}$ and $\mathbb{P}(\xi=k)=1$ for different values of $k$.}
    \label{fig:comparison}
\end{figure}

Simulations for the case represented in Figure \ref{fig:comparison} are simple because we assumed $\mathbb{P}(\xi=k)=1$. Although this can be extended assuming different laws of $\xi$, in some cases the computational cost is high so we can appeal to the theory of branching processes. Indeed, we can simply use the fact that the number of propagating individuals at the $nth$-level of the tree is given by $m^n$ where $m$ is the mean of the number of contacts that propagate the information from a given individual. In Figure \ref{fig:comparisongama} we present the results coming from this analysis for the $(p,\xi)$-MPE model (left side) $vs$ the $(p,\xi)$-MAE model (right side), assuming $\xi\sim Zipf(\gamma)$ for different values of $\gamma$, taking different values of $p$, and using log-linear scales for the results. 

\newpage
\begin{figure}[h!]
    \centering
    \subfigure[$\gamma=1.1$.]{\includegraphics[scale=0.43]{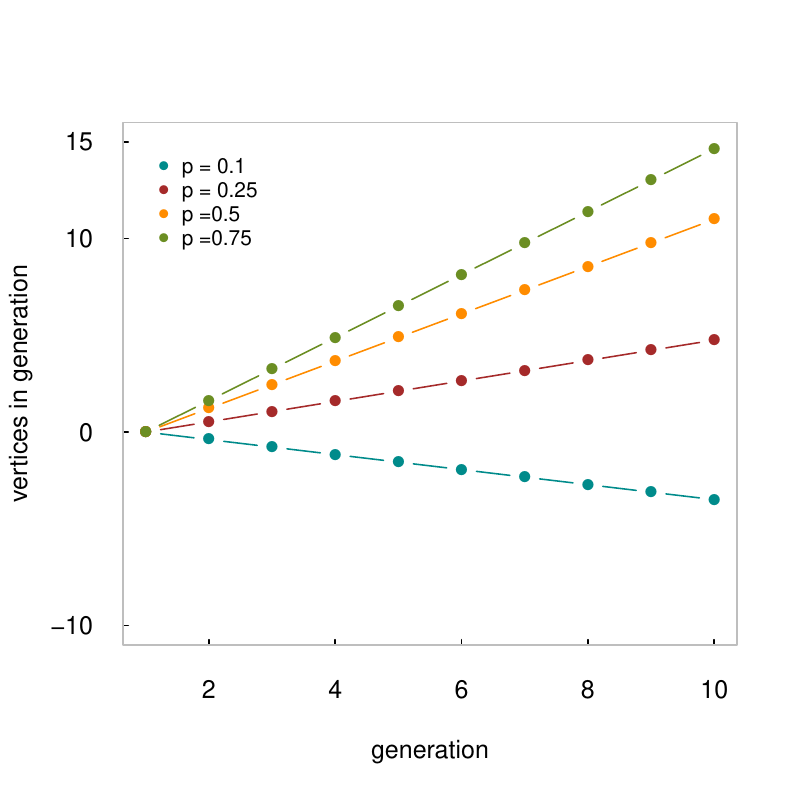}\includegraphics[scale=0.43]{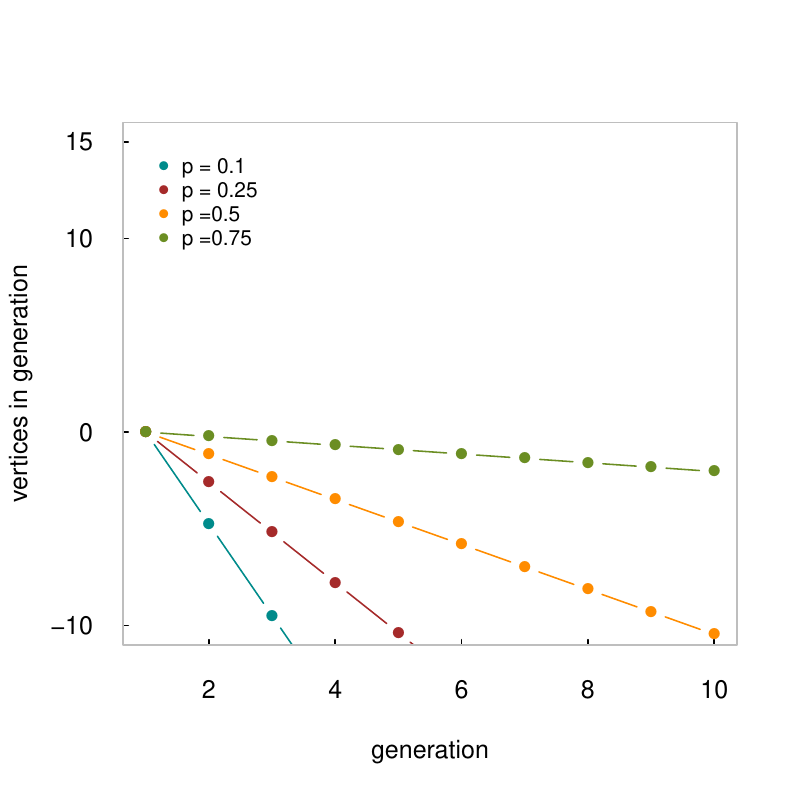}}
    \subfigure[$\gamma=1.05$.]{\includegraphics[scale=0.43]{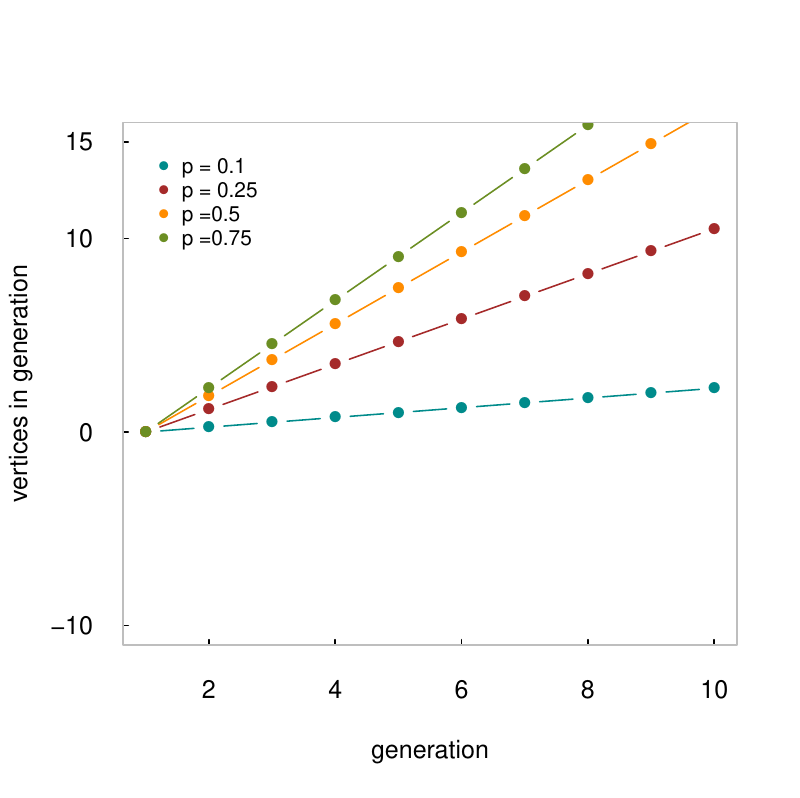}\includegraphics[scale=0.43]{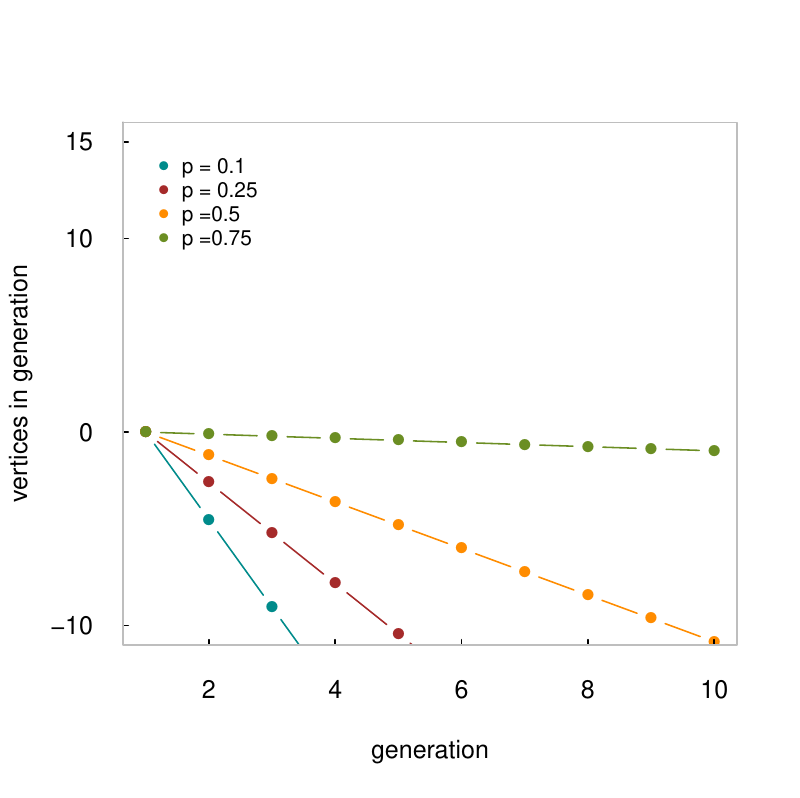}}
    \subfigure[$\gamma=1.03$.]{\includegraphics[scale=0.43]{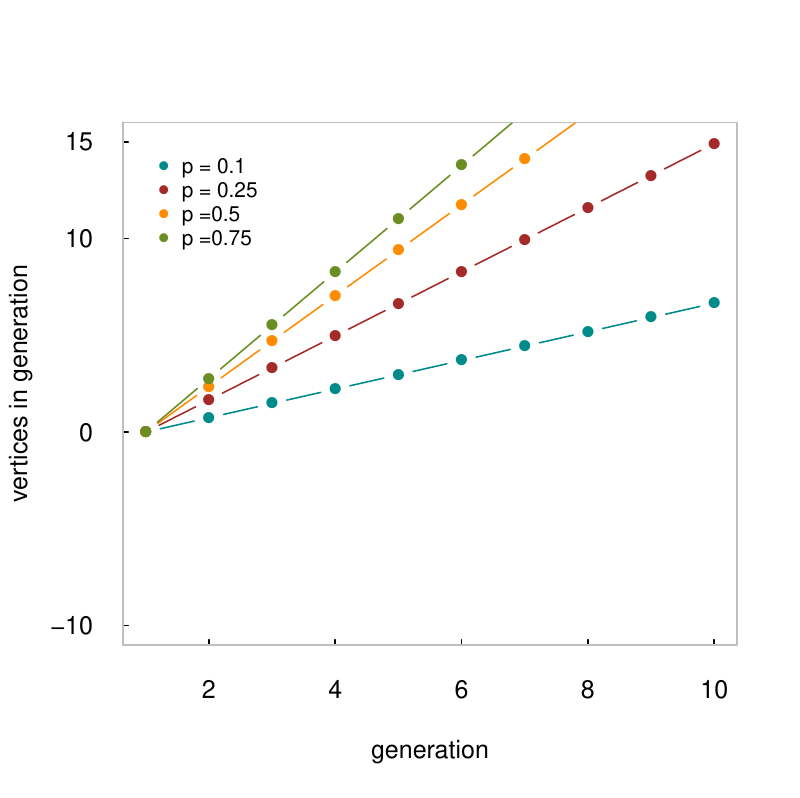}\includegraphics[scale=0.43]{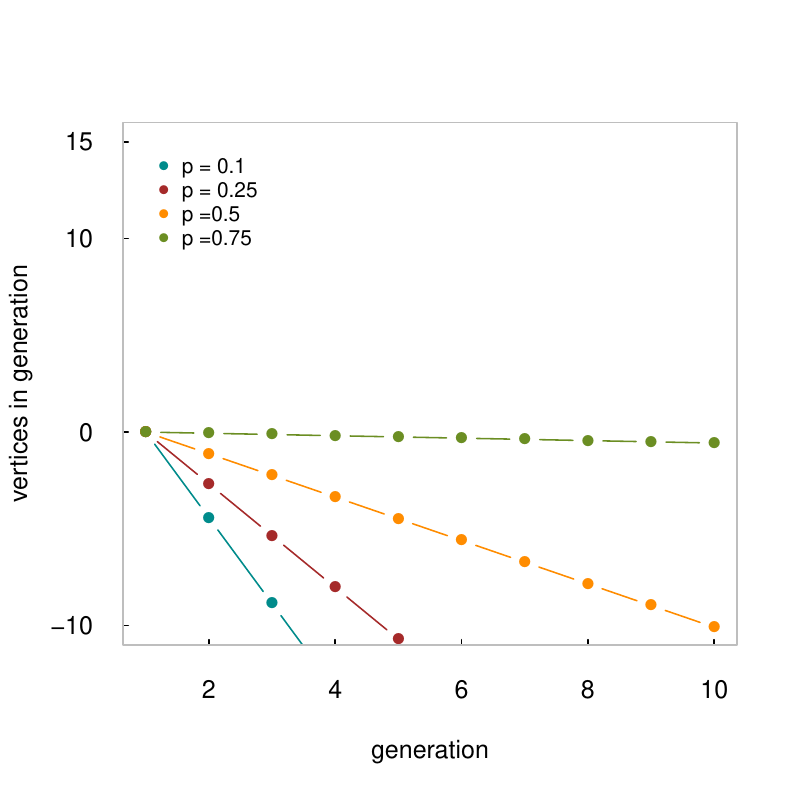}}
    \caption{Evolution of the propagation for the $(p,\xi)$-MPE model (left side) $vs$ the $(p,\xi)$-MAE model (right side). We exhibit the mean values obtained from $8\times 10^6$ simulations of the processes for $p\in\{0.1,0.25,0.5,0.75\}$ and $\xi\sim Zipf(\gamma)$ for different values of $\gamma$.}
    \label{fig:comparisongama}
\end{figure}

As all our comparison results show the propagation of a misinformation reduces drastically if we assume an active environment. This leave an important message, it is not enough just not propagate a false information in order to contribute with the stopping of the propagation. It is the fact of inhibiting the propagator, or denouncing the propagation implying its blocking, which makes the difference. We caught this behaviour through two simple probabilistic models based on the theory of branching processes. We emphasize that such a theory is rich enough to improve these models to incorporate new scenarios of interest so we expect that our results can inspire new research in this direction.

\section{Proofs}

Our models are branching processes so we can appeal to well-known results to obtain ours. Let us start with some results of branching processes.

\begin{prop} \label{branching}
Let $\{Z_n\}_{n \geq 0}$ be a branching process with $Z_0=1$. Assume that the offspring distribution $\{p_k\}_{k\in\mathbb{N}\cup\{0\}}$,  with $p_k := \mathbb{P}(Z_1 = k)$, is such that $p_0 > 0$, $p_0+p_1 < 1$, and it has the generating function $g(s) =  \sum_k s^kp_k$. Let  $\mu := \mathbb{E}(Z_1) = g^{\prime}(1)$, the mean number of offspring per individual and let $\psi$ be the probability that, starting with a single individual, the population ever dies out. In addition, we write $\mathbb{P}(E) = \psi$. That is, $E$ is the event that the process dies out. \\
(a)Then 
\begin{enumerate} [label=\roman*.]
    \item $\psi$ is the smallest positive number satisfaying $g(s) = s.$
    \item $\psi = 1$ if, and only if, $\mu \leq 1$.
\end{enumerate}
(b) Let $Y$ be a total progeny (or total population size). That is,
\[ Y := \sum_{n=0}^{\infty}Z_n.
\]

 If $\mu < 1$, then
    \[ \mathbb{P}(Y = k) = \frac{1}{k} \mathbb{P}(X_1 + X_2 + \ldots + X_k = k-1),
    \]
    where $X_1, X_2, \ldots, X_k$ are independent random variables identically distributed as $Z_1$. Also,
    \[ \mathbb{E}(Y) = \frac{1}{1-\mu}.
    \]
   (c) Let $T$ be the extinction time of the branching process. That is, 
\[ T := \min \{n > 0, Z_n =0  \}.
\]
Suppose that the offspring distribution $\{p_k\}_k$ has a fractional linear generating function. That is, assume that $g$ has the form
\[ g(s) = 1 - \mu(1-c) + \frac{\mu (1-c)^2s}{1-cs}, 0 \leq s \leq 1 \textrm { and } \mu (1-c) < 1. 
\]
Suppose that $\mu \neq 1$. We have that
\[ \mathbb{P}( T \leq n) = \frac{s_0(1- \mu^n)}{s_0 - \mu^n}, n \geq 1
\]
where $s_0$ is the non-negative solution to the equation $g(s) = s$ different from 1, given by
\[ s_0 = s_0(c) = \frac{1- \mu(1-c)}{c}.
\]
Furthermore, $if \mu < 1$, 
\[ \mathbb{E}(T) = (s_0-1) \sum_{n=0}^{\infty}\frac{\mu^n}{s_0 - \mu^n}.
\]
(d) Let $T$ be the extinction time of the branching process. 
Suppose that $\mu < 1$ and for all $s \in [0,1]$, $g^\prime(1)g^{\prime \prime}(s) - g^\prime(s)g^{\prime \prime}(1) \geq 0$.
Write
\[ l := \frac{2g^\prime(1) + g^{\prime \prime}(1) - 2[g^\prime(1)]^2}{g^{\prime \prime}(1)} \textrm { and } u := \frac{g(0)g^\prime(1)}{g^\prime(1) +g(0) -1}.
\]
Then 
\[ \frac{l(1 - \mu^n)}{l - \mu^n} \leq \mathbb{P}(T \leq n) \leq \frac{u(1 - \mu^n)}{u - \mu^n}
\]
\[ [u -1] \sum_{n=0}^{\infty}\frac{\mu^n}{u - \mu^n} \leq \mathbb{E}(T) \leq [l -1] \sum_{n=0}^{\infty}\frac{\mu^n}{l - \mu^n}.
\]
\end{prop}
\begin{proof}[Proof of Proposition~\ref{branching}]
The proof of $(a)$ can be found in \cite[Theorem 6.1]{harris}. Item $(b)$ follows from \cite[Theorem, page 682]{dwass}.\\

To prove (c) first note that
\[ \mathbb{P}(T \leq n) = \mathbb{P}(Z_n = 0) = g^m(0)
\]
where $g^m(0)$ is the $m-$th compositions of $g$. The last equality follows from \cite[Theorem 4.1, page 5]{harris}. Now, we refer the reader to \cite[Definition 2]{hwang} and  \cite[expression (2.1)]{agresti/1974}, for a suitable definition of a fractional linear generating function. Then, the results follow from \cite[expression (3.1)]{agresti/1974}.

To prove $(d)$ we appeal to arguments of \cite{agresti/1974}
for deriving bounds for the tail of such a random time. The main approach of \cite{agresti/1974} consists of
reducing the problem of deriving these bounds to a problem involving the analysis of a family of the fractional linear generating functions (f.l.g.f.). whose iterates were obtained in the proof of $(c)$.  Our first task is to obtain two f.l.g.f. $f_l(s)$ and $f_u(s)$ such that
\[
f_l(s) \leq g(s) \leq f_u(s), 0 \leq s \leq 1.
\]
In order to do it we apply \cite[Theorem, p. 450]{hwang}, verifying its corollary of p. 451. Then, $f_l(s)$ and $f_u(s)$ are such that
\[ c_l = \frac{g^{\prime \prime}(1)}{2g^{\prime}(1) + g^{\prime \prime}(1)} \textrm { and } c_u = \frac{g^{\prime}(1) +g(0) -1}{g^{\prime \prime}(1)}.
\]
Note that the inequalities hold also for the $m-$th compositions of these functions. Thus,
\[
f^m_l(0) \leq g^m(0) \leq f^m_u(0).
\]
Using $(c)$ and \cite[Section 4]{agresti/1974} we finish the proof.

\end{proof}

Now the main idea of our proofs is the fact that our misinformation models are branching processes. 

\begin{proof} [\bf Proof of Theorem~\ref{thm:model1}]
We associate the $(p,\xi)$-MPE model with a branching process $\{Z^1_n, n \geq 0 \}$ with offspring distribution given by
\[ \mathbb{P}(Z^1_1 = j) = \sum_{n=j}^{\infty}\mathbb{P}(Z^1_1 = j | \xi = n)\mathbb{P}(\xi = n) = \sum_{n=j}^{\infty} \binom{n}{j}p^j(1-p)^{n-j}\mathbb{P}(\xi = n), j \geq 0.
\]
Note that $Z^1_1$ has an expected value given by
\[ \mu_1 := \mathbb{E}(Z^1_1) =  \sum_{j=1}^{\infty} j \sum_{n=j}^{\infty} \binom{n}{j}p^ j(1-p)^{n-j}\mathbb{P}(\xi = n)  =  \sum_{n=1}^{\infty} \sum_{j=0}^{n}j\binom{n}{j}p^ j(1-p)^{n-j}\mathbb{P}(\xi = n)  = p\mathbb{E}(\xi)
\]
and a generating probability function
\begin{align*} 
g_1(s) &= \mathbb{E}(s^{Z^1_1}) = \sum_{j=0}^{\infty} s^j \sum_{n=j}^{\infty} \binom{n}{j}p^ j(1-p)^{n-j}\mathbb{P}(\xi = n)  =  \sum_{n=0}^{\infty}  \sum_{j=0}^{n} \binom{n}{j}\left(\frac{sp}{1-p}\right)^j(1-p)^{n}\mathbb{P}(\xi = n) \\  & = \sum_{n=0}^{\infty}  \sum_{j=0}^{n} (1-p)^{n} \mathbb{P}(\xi = n) \binom{n}{j}\left(\frac{sp}{1-p}\right)^j = \mathbb{E} \left [\left( sp+1-p\right)^{\xi} \right ] = G_{\xi}(sp+1-p).
\end{align*}
The result follows from the application of Proposition \ref{branching} $(a)$.
\end{proof}

\smallskip 
\begin{proof} [\bf Proof of Corollary~\ref{cor:model1tfc}]
Just apply Theorem \ref{thm:model1} with
\[ p\mathbb{E}(\xi) = pk \textrm { and } G_{\xi}(sp+1-p) = (sp+1-p)^k.
\]
\end{proof}

\smallskip 
\begin{proof} [\bf Proof of Corollary~\ref{cor:model1tfzeta}]
Just apply Theorem \ref{thm:model1} with
\[ p\mathbb{E}(\xi) = p \sum_{j=1}^{\infty} j\frac{j^{-\gamma}}{\zeta(\gamma)} = \frac{p\zeta(\gamma -1)}{\zeta(\gamma)} 
\]
 and 
\[ G_{\xi}(sp+1-p) = \sum_{j=1}^{\infty}(sp+1-p)^j\frac{j^{-\gamma}}{\zeta(\gamma)} = \frac{1}{\zeta(\gamma)}Li_{\gamma}(sp+1-p),
\]
where 
\[Li_{a}(x) = \sum_{j=1}^{\infty}\frac{x^j}{j^a}, |x| < 1
\]
is the polylogarithm function also known as the Jonquière's function.
\end{proof}

\smallskip 
\begin{proof} [\bf Proof of Theorem~\ref{thm:model1alcance}]
In the proof of Theorem~\ref{thm:model1} we obtain that process $\{Z^1_n, n \geq 0 \}$ has generating probability function $g_1(s) = G_{\xi}(sp+1-p)$. After some calculations we get 
\[  g_1^{\prime}(s)   = p\mathbb{E} \left(\xi(sp+1-p)^{\xi -1} \right),~~ g_1^{\prime \prime}(s)  = p^2\mathbb{E} \left(\xi(\xi-1)(sp+1-p)^{\xi -2} \right),~~ g_1^{\prime}(1) = p\mathbb{E} \left(\xi \right),  
\]
\[  g_1^{\prime \prime}(1) = p^2\mathbb{E} \left(\xi(\xi-1) \right) \textrm { and } g_1(0) = G_{\xi}(1-p).
\]
Now note that $\mathcal{T}_h+1 = T$ where $T = \min\{n \geq 0, Z^1_n = 0 \}$.

Then, we apply Proposition \ref{branching} $d)$ with
\[ l = \frac{2 g_1^{\prime}(1) + g_1^{\prime \prime} (1)- 2[ g_1^{\prime}(1)]^2}{ g_1^{\prime \prime}(1) } = \frac{2p\mathbb{E}(\xi) +p^2\mathbb{E} \left(\xi(\xi-1) \right) -2 \left( p\mathbb{E}(\xi) \right )^2}{p^2\mathbb{E} \left(\xi(\xi-1) \right)}
\]
and
\[ u =  \frac{g_1(0) g_1^{\prime}(1)}{ g_1^{\prime}(1) +g_1(0) -1} = \frac{G_{\xi}(1-p)p\mathbb{E}(\xi)}{p\mathbb{E}(\xi) + G_{\xi}(1-p) - 1}
\]
\end{proof}

\begin{proof} [\bf Proof of Corollary~\ref{cor:model1alcance}]
Since $\mathbb{E}(v(\xi)) = v(k)$ follows that
\begin {align*} U(s) &=\mathbb{E}(\xi) [\mathbb{E} \left (\xi (sp+1-p)^{\xi-2} (\xi(2+sp-p) -1) \right)] - \mathbb{E} \left ( \xi^2  \right )\mathbb{E} \left(\xi (sp+1-p)^{\xi-1} \right ) \\
U(s) &= k^2(k-1)(sp+1-p)^{k-2} \geq 0, s \in [0;1].
\end{align*}

Just apply Theorem \ref{thm:model1alcance} with 
\[ G(s) = \left (G_{\xi}(sp+1-p) \right )^j = \left ((sp+1-p)^k \right)^j.  
\]

\end{proof}

\smallskip
\begin{proof} [\bf Proof of Corollary~\ref{cor:model1alcancezeta}]
We have that
\[
G_{\xi}(s) = \sum_{j=1}^{\infty}s^j \frac{j^{-\gamma}}{\zeta(\gamma)} = \frac{1}{\zeta(\gamma)}Li_{\gamma}(s).
\]
Hence we apply Theorem \ref{thm:model1alcance}.
\end{proof}

\smallskip
\begin{proof} [\bf Proof of Theorem~\ref{thm:model1alcance-size}]
First observe that if $X_1, X_2, . . . , X_j$ are independent random variables identically distributed as $Z_1$ then
\[ \mathbb{P}(X_1 + X_2 + \ldots + X_j = j - 1) = \frac{ G^{(j-1)}(0)}{(j-1)!}, j=1,2, \ldots
\]
where
\[
 G(s) = \left ( g_1(s)  \right )^j, \textrm { with } g_1(s) = G_{\xi}(sp+1-p).
 \]
 The result follows from the Proposition \ref{branching} $b)$.
\end{proof}

\smallskip
\begin{proof} [\bf Proof of Corollary~\ref{cor:model1alcance-size}]
Immediate application of Theorem \ref{thm:model1alcance-size} with $G_{\xi}(s) = s^k.$ 
\end{proof}

\smallskip
\begin{proof} [\bf Proof of Corollary~\ref{cor:model1alcancezeta-size}]
Immediate application of Theorem \ref{thm:model1alcance-size} with $G_{\xi}(s) =  \frac{1}{\zeta(\gamma)}Li_{\gamma}(s)$.
\end{proof}

\smallskip
\begin{proof} [\bf Proof of Theorem~\ref{thm:model2}]
We associate the $(p,\xi)$-MAE model with a branching process $\{Z^1_n, n \geq 0 \}$ with offspring distribution given by
\begin{align*}
\mathbb{P}(Z^2_1 = k) &= \mathbb{P}(Z^1_1 = k | \xi = k)\mathbb{P}(\xi = k)  + \sum_{n=k+1}^{\infty}\mathbb{P}(Z^1_1 = k | \xi = n)\mathbb{P}(\xi = n) \\
&= p^k\mathbb{P}(\xi = k) + \sum_{n=k+1}^{\infty}(1-p)p^k \mathbb{P}(\xi = n) = p^k\mathbb{P}(\xi = k) + (1-p)p^k \mathbb{P}(\xi \geq k+1) \\
&= p^k [\mathbb{P}(\xi \geq k) - p\mathbb{P}(\xi > k) ]
\end{align*}
$Z^2_1$ has expected value given by
\begin{align*} \mu_G := \mathbb{E}(Z^2_1) &=   \sum_{k=1}^{\infty} kp^k [\mathbb{P}(\xi \geq k) - p\mathbb{P}(\xi > k) ] = \sum_{k=1}^{\infty} kp^k \mathbb{P}(\xi \geq k) - p\sum_{k=1}^{\infty} kp^k [\mathbb{P}(\xi \geq k) - \mathbb{P}(\xi = k)]\\
& = (1-p)\sum_{k=1}^{\infty} kp^k \mathbb{P}(\xi \geq k) +p \mathbb{E} \left ( \xi p^{\xi} \right )
\end{align*}
and generating probability function
\begin{align*} 
g_2(s) &= \mathbb{E}(s^{Z^2_1}) = \sum_{k=0}^{\infty} s^k \left [ p^k [\mathbb{P}(\xi \geq k) - p\mathbb{P}(\xi > k) ] \right ] = \sum_{k=0}^{\infty} (sp)^k \left [ \mathbb{P}(\xi \geq k) - p \mathbb{P}(\xi \geq k) + p \mathbb{P}(\xi = k) \right ] \\
&= (1-p) ] \sum_{k=0}^{\infty} (sp)^k \mathbb{P}(\xi \geq k) + pG_{\xi}(sp).
\end{align*}
The result follows from the application of the Proposition \ref{branching} $(a)$.
\end{proof}

\smallskip
\begin{proof} [\bf Proof of Corollary~\ref{cor:model2tfc}]
Applying the Theorem \ref{thm:model2} we have
\[ \frac{p(1 - (1 + k)p^k + k p^{k+1})}{(1-p)^2} = \sum_{x=1}^{k}xp^x = \sum_{x=1}^{\infty}xp^x \mathbb{P}(\xi \geq x) > \frac{1-p\mathbb{E}(\xi p^{\xi})}{1-p} = \frac{1-kp^{k+1}}{1-p}, 
\]
\[
\frac{1-(sp)^{k+1}}{1-sp} = \sum_{x=0}^{k}(sp)^x  = \sum_{x=0}^{\infty}(sp)^x \mathbb{P}(\xi \geq x) = \frac{s-p G_{\xi}(sp)}{1-p} = \frac{s-p(sp)^k}{1-p} 
\]
\end{proof}

\smallskip
\begin{proof} [\bf Proof of Corollary~\ref{cor:model2tfzeta}]
Applying the Theorem \ref{thm:model2} we have
\[ \frac{1}{\zeta(\gamma)}\sum_{x=1}^{\infty}xp^x \zeta(\gamma, x) = \sum_{x=1}^{\infty}xp^x \mathbb{P}(\xi \geq x) > \frac{1-p\mathbb{E}(\xi p^{\xi})}{1-p} = \frac{1 - \frac{p}{\zeta(\gamma)}Li_{\gamma-1}(p)}{1-p} = \frac{\zeta(\gamma) -pLi_{\gamma -1}(p)}{\zeta(\gamma)(1-p)}, 
\]
\[
\frac{1}{\zeta(\gamma)} \sum_{x=1}^{\infty}(sp)^x\zeta(\gamma,x) = \sum_{x=0}^{\infty}(sp)^x \mathbb{P}(\xi \geq x) = \frac{s-p G_{\xi}(sp)}{1-p} = \frac{s -\frac{p}{\zeta(\gamma)}Li_{\gamma}(sp)}{1-p} = \frac{s \zeta(\gamma) -p Li_{\gamma}(sp)}{\zeta(\gamma)(1-p)}
\]
\end{proof}

\smallskip
\begin{proof} [\bf Proof of Theorem~\ref{thm:model2alcance}]
In the proof of Theorem~\ref{thm:model2} we obtain that process $\{Z^2_n, n \geq 0 \}$ has generating probability function 
\[g_2(s) = (1-p) ] \sum_{k=0}^{\infty} (sp)^k \mathbb{P}(\xi \geq k) + pG_{\xi}(sp).
\]
After some calculations we get 
\begin{align*} \mu_G(s) =g_2^{\prime}(s)  &= \frac{(1-p)}{s}\sum_{x=1}^{\infty}x(ps)^x \mathbb{P}(\xi \geq x) + \frac{p}{s}\mathbb{E}(\xi (ps)^{\xi})    ,\\ B_G(s) = g_2^{\prime \prime}(s) &= \frac{(1-p)}{s^2}\sum_{x=2}^{\infty}x(x-1)(ps)^x \mathbb{P}(\xi \geq x) + \frac{p}{s^2}\mathbb{E}(\xi(\xi-1)(ps)^{\xi}) ,\\ \mu_G = g_2^{\prime}(1)  &= (1-p)\sum_{x=1}^{\infty}xp^x \mathbb{P}(\xi \geq x) + p\mathbb{E}(\xi p^{\xi}) \\
B_G = g_2^{\prime \prime}(1) &= (1-p)\sum_{x=2}^{\infty}x(x-1)p^x \mathbb{P}(\xi \geq x) + p\mathbb{E}(\xi(\xi-1)p^{\xi}) \textrm { and } \\ g_2(0) &= 1-p +p\mathbb{P}(\xi = 0) = 1-p\mathbb{P}(\xi \neq 0) .
\end{align*}
Now note that $\mathcal{T}_h+1 = T$ where $T = \min\{n \geq 0, Z^1_n = 0 \}$.

Then, we apply Proposition \ref{branching} $d)$ with
\[ l = \frac{2 g_2^{\prime}(1)  +  g_2^{\prime \prime} (1) - 2[ g_2^{\prime}(1)]^2}{ g_2^{\prime \prime}(1)} = \frac{2\mu_G(1-\mu_G) +B_G}{B_G},
\]
\[ u =  \frac{g_2(0) g_2^{\prime}(1) }{ g_2^{\prime}(1) +g_2(0) -1} = \frac{[1-p\mathbb{P}(\xi \neq 0)]\mu_G}{\mu_G - p(\mathbb{P}(\xi \neq 0))},
\]

\end{proof}

\smallskip
\begin{proof} [\bf Proof of Corollary~\ref{cor:model2alcancek}] 
We have that
\[ \mu_G(s) = \frac{(1-p)}{s}\sum_{x=1}^{\infty}x(ps)^x \mathbb{P}(\xi \geq x) + \frac{p}{s}\mathbb{E}(\xi (ps)^{\xi}) = \frac{1}{s} \left [ (1-p) \sum_{x=1}^{k}x(ps)^x +kp(sp)^k \right],
\]
\begin{align*} B_G(s) &=  \frac{(1-p)}{s^2}\sum_{x=2}^{\infty}x(x-1)(ps)^x \mathbb{P}(\xi \geq x) + \frac{p}{s^2}\mathbb{E}(\xi(\xi-1)(ps)^{\xi})\\ &= \frac{1}{s^2} \left [ (1-p) \sum_{x=2}^{k}x(x-1)(ps)^x +k(k-1)p(sp)^k \right], 
\end{align*}
\[ \mu_G = (1-p)\sum_{x=1}^{k}xp^x +pkp^k = \frac{p(1-p^k)}{1-p}, \mathbb{P}(\xi \neq 0) = 1 \textrm { and }
\] 
\[ B_G = (1-p)\sum_{x=2}^{k}x(x-1)p^x +pk(k-1)p^k = \frac{2p[(k-1)p^{k+1}-kp^k+p]}{(1-p)^2}.
\]
To finished the proof just apply Theorem \ref{thm:model2alcance}.
\end{proof}

\smallskip
\begin{proof} [\bf Proof of Corollary~\ref{cor:model2alcancez}] 
After some calculations, we get
\[
\dfrac { \partial h(s,p,\gamma)}{\partial s} = \frac{1}{s} \sum_{x=1}^{\infty}x(sp)^{x-1}\zeta(\gamma,x),
\dfrac { \partial^2 h(s,p,\gamma)}{\partial s^2} = \frac{1}{s^2} \sum_{x=2}^{\infty}x(x-1)(sp)^{x-2}\zeta(\gamma,x),
\]
\[\mu_G(s) = \frac{(1-p)}{s}\sum_{x=1}^{\infty}x(ps)^x \mathbb{P}(\xi \geq x) + \frac{p}{s}\mathbb{E}(\xi (ps)^{\xi}) = \frac{1}{s \zeta(\gamma)}\left [(1-p)\dfrac { \partial h(s,p, \gamma)}{\partial s} + p Li_{\gamma -1}(ps) \right ],
\]
\begin{align*} B_G(s) &=  \frac{(1-p)}{s^2}\sum_{x=2}^{\infty}x(x-1)(ps)^x \mathbb{P}(\xi \geq x) + \frac{p}{s^2}\mathbb{E}(\xi(\xi-1)(ps)^{\xi})\\ &=  \frac{1}{s^2\zeta(\gamma)} \left [(1-p)\dfrac { \partial^2 h(s,p, \gamma)}{\partial s^2} + p [Li_{\gamma -2}(ps) - Li_{\gamma -1}(ps)] \right ], 
\end{align*}
\[ \mu_G = (1-p)\sum_{x=1}^{\infty}xp^x\sum_{j=x}^{\infty}\frac{j^{-\gamma}}{\zeta(\gamma)} + p\sum_{j=1}^{\infty}jp^j\frac{j^{-\gamma}}{\zeta(\gamma)} = \frac{1}{\zeta(\gamma)}\left [(1-p)\dfrac { \partial h(1,p, \gamma)}{\partial s} + p Li_{\gamma -1}(p) \right ] 
\]
and
\begin{align*} B_G &= (1-p)\sum_{x=2}^{\infty}x(x-1)p^x\sum_{j=x}^{\infty}\frac{j^{-\gamma}}{\zeta(\gamma)} + p\sum_{j=2}^{\infty}j(j-1)p^j\frac{j^{-\gamma}}{\zeta(\gamma)} \\ B_G &= \frac{1}{\zeta(\gamma)} \left [(1-p)\dfrac { \partial^2 h(1,p, \gamma)}{\partial s^2} + p [Li_{\gamma -2}(p) - Li_{\gamma -1}(p)] \right ]. 
\end{align*}
To finished the proof just apply Theorem \ref{thm:model2alcance}.
\end{proof}

\smallskip
\begin{proof} [\bf Proof of Theorem~\ref{thm:model2alcance-size}] 
First observe that if $X_1, X_2, . . . , X_j$ are independent random variables identically distributed as $Z_1$ then
\[ \mathbb{P}(X_1 + X_2 + \ldots + X_j = j - 1) = \frac{ G^{(j-1)}(0)}{(j-1)!}, j=1,2, \ldots
\]
where
\[
 G(s) = \left ( g_2(s)  \right )^j, \textrm { with } g_2(s) = (1-p)\sum_{x=0}^{\infty} (sp)^x \mathbb{P}(\xi \geq x) + pG_{\xi}(sp).
 \]
 The result follows from the Proposition \ref{branching} $(b)$.
\end{proof}

\smallskip
\begin{proof} [\bf Proof of Corollary~\ref{cor:model2alcancek-size}] 
We apply Theorem \ref{thm:model2alcance-size} with
\[ G(s) = \left ( (1-p)\sum_{x=0}^{\infty} (sp)^x \mathbb{P}(\xi \geq x) + pG_{\xi}(sp)\right )^j = \left ( \frac{p(1-s)(ps)^k +1-p}{1-ps}  \right )^j.
\]
\[ \mathbb{E}(\mathcal{T}_s) = \left ( 1-pkp^k -(1-p)\sum_{x=1}^{k}xp^x  \right )^{-1} = \frac{1-p}{1-2p+p^{k+1}}
\]
\end{proof}

\smallskip
\begin{proof} [\bf Proof of Corollary~\ref{cor:model2tamanhoz}] 
We apply Theorem \ref{thm:model2alcance-size} with
\[ G(s) = \left ( (1-p)\sum_{x=0}^{\infty} (sp)^x \mathbb{P}(\xi \geq x) + pG_{\xi}(sp)\right )^j = \left ( \frac{(1-p)h(s,p, \gamma) + pLi_{\gamma}(sp)}{\zeta(\gamma)}\  \right )^j.
\]

\[  \left (1-p\mathbb{E}(\xi p^{\xi}) -(1-p)\sum_{x=1}^{\infty} xp^x \mathbb{P}(\xi \geq x) \right)^{-1} = \frac{\zeta(\gamma)}{\zeta(\gamma) - p Li_{\gamma -1}(p) -(1-p) \dfrac { \partial h(1,p, \gamma)}{\partial s}}
\]
\[ \left (1- \frac{(1-p)}{s}\sum_{x=1}^{\infty}x( \alpha p)^x \mathbb{P}(\xi \geq x)- \frac{p}{\alpha}\mathbb{E}(\xi (\alpha p)^{\xi}) \right )^{-1} = \frac{\alpha \zeta(\gamma)}{\alpha \zeta(\gamma) -(1-p)\dfrac { \partial h(\alpha,p,\gamma)}{\partial s} - p Li_{\gamma -1}(\alpha p)}.
\]

\end{proof}

\smallskip
\begin{proof} [\bf Proof of Theorem~\ref{thm:comp}] 
Just calculate
\[ \mathbb{E}_B(\mathcal{T}_s) - \mathbb{E}_G(\mathcal{T}_s) = \frac{1}{1-p\mathbb{E}(\xi)} - \left (1-p\mathbb{E}(\xi p^{\xi}) -(1-p)\sum_{x=1}^{\infty} xp^x \mathbb{P}(\xi \geq x) \right)^{-1}
\]

\end{proof}

\section*{Acknowledgements}
Part of this work was carried out during a stay of P.M.R. and L.M.G. at CMCC/UFABC, and a visit of L.M.G. to CCEN/UFPE. The authors are grateful to these institutions for their hospitality and support. Part of this work has been supported by FACEPE - Fundação de Amparo à Ciência e Tecnologia do Estado de Pernambuco (Grant APQ-1341-1.02/22), and
FAPESP - Fundação de Amparo à Pesquisa do Estado de São Paulo (Grant 2017/10555-0, Grant 2019/23078-1).

\end{document}